\documentclass[12pt,onecolumn,twoside,draftcls]{IEEEtran}

\date{}

\usepackage{graphics,epsfig,slashbox}
\usepackage{graphicx}
\usepackage{subfigure}
\usepackage{amssymb, amsmath}
\usepackage{multirow}
\usepackage{algorithm}
\usepackage{algorithmic}
\usepackage{url}
\usepackage{afterpage}
\usepackage{color}
\usepackage{epstopdf}

\newcommand{\ignore}[1]{}

\newcommand{\x}{{\bf x}}
\newcommand{\y}{{\bf y}}
\newcommand{\z}{{\bf z}}

\newcommand{\n}{{\bf {n}}}

\newcommand{\argmax}{\operatornamewithlimits{argmax}}

\newtheorem{lemma}{Lemma}
\newtheorem{property}{Property}

\begin{document}
\title{\LARGE On the feasibility of beamforming in millimeter wave communication systems with multiple antenna arrays
}

\author{\IEEEauthorblockN{Jaspreet Singh and Sudhir Ramakrishna}

\IEEEauthorblockA{Samsung Research America - Dallas, TX, USA \\
Email: \{jaspreet.s, sudhir.r\}@samsung.com \\
\vspace{-10mm}
\thanks{
A preliminary version \cite{GC14}  of this paper has been accepted for publication at IEEE Globecom, 2014 .
}
}
}

\maketitle \vspace{-7mm}

\vspace{-3mm}
\begin{abstract}
\vspace{-3mm}
The use of the millimeter (mm) wave spectrum for next generation (5G) mobile communication has gained significant attention recently. The small carrier wavelengths at mmwave frequencies enable synthesis of compact antenna arrays, providing beamforming gains that compensate the increased propagation losses. In this work, we investigate the feasibility of employing multiple antenna arrays (at the transmitter and/or receiver) to obtain diversity/multiplexing gains in mmwave systems, where each of the arrays is capable of beamforming independently. Considering a codebook based beamforming system (the set of possible beamforming directions is fixed {\emph a priori}, e.g., to facilitate limited feedback), we observe that the complexity of \emph{jointly} optimizing the beamforming directions across the multiple arrays is highly prohibitive, even for very reasonable system parameters. To overcome this bottleneck, we develop reduced complexity algorithms for optimizing the choice of beamforming directions, premised on the sparse multipath structure of the mmwave channel. Specifically, we reduce the cardinality of the joint beamforming search space, by restricting attention to a small set of dominant candidate directions. To obtain the set of dominant directions, we develop two complementary approaches: (a) based on computation of a novel spatial power metric; a detailed analysis of this metric shows that, in the limit of large antenna arrays, the selected candidate directions approach the channel's dominant angles of arrival and departure, and (b) precise estimation of the channel's  (long-term) dominant angles of arrival, exploiting the correlations of the signals received across the different receiver subarrays. Our methods enable a drastic reduction of the optimization search space (a factor of 100 reduction), while delivering close to optimal performance, thereby indicating the potential feasibility of achieving diversity and multiplexing gains in mmwave systems. \looseness-1
\end{abstract}

\vspace{-2mm}
\IEEEpeerreviewmaketitle

\vspace{-7mm}
\section{Introduction}

The use of the millimeter (mm) wave band for next generation (5G) mobile communication has gained considerable attention recently \cite{Pi_MMW, Rangan}. Vast amounts of spectrum (both licensed and unlicensed) are available in the mmwave band (typically considered to be 30-300 GHz), making it attractive for high data rate communication. While propagation losses in the mmwave band are higher compared to those at lower microwave frequencies used currently for mobile communication, the smaller carrier wavelengths for mmwave frequencies also imply that more antennas can be packed in a relatively small area, making it feasible to design compact high-gain antenna arrays that can compensate for the increased propagation losses.\looseness-1

In current systems employing multiple antennas (i.e., current MIMO systems, e.g., 3GPP LTE \cite{36_211}), beamforming (or, precoding) is performed at baseband (BB), and the precoder outputs are fed into the different transmit antennas using a separate radio frequency (RF) chain (implied to include the upconversion components and the power amplifier) for each antenna. For mmwave systems employing large antenna arrays, such an architecture is considered infeasible, given the prohibitive cost of the large number of RF chains and mixed signal components (D/A and A/D converters) \cite{ICUWB09}. Rather, beamforming may be performed at the RF level using a set of analog phase shifters, which limits the number of required RF chains to one.\looseness-1

To obtain diversity/spatial multiplexing gains in such systems, a natural strategy is to consider the use of a transmitter (Tx) and/or receiver (Rx) antenna array consisting of multiple \emph{subarrays}, where each of the subarrays is capable of independent electronic beam steering using RF phase shifters (concept first introduced in \cite{Torkildson}, in context of line-of-sight mmwave MIMO systems), see Fig. \ref{fig:CSI}. In essence, each subarray emulates a \emph{virtual} antenna (that is capable of directional transmission) in the sense of current MIMO systems, and spatial multiplexing of $L$ data streams can be supported, in principle, using at least $L$ subarrays at both the Tx and the Rx. In addition to the RF beamforming at each subarray, a baseband precoder may also be employed at the Tx (as in current systems; with each subarray being analogous to one antenna in current systems) to process the data to be sent on different streams, providing an additional level of flexibility on top of only the phase shift operations performed at RF. The overall precoding operation may then be referred to as \emph{hybrid} (mix of analog and digital) precoding.\looseness-1

We consider the preceding \emph{array-of-subarrays} architecture, and study the corresponding hybrid precoder optimization problem, involving a joint optimization over the choice of the Tx BB precoder, and the Tx/Rx RF precoders (i.e., the Tx/Rx RF beamforming directions at the different subarrays). In this work, we consider codebook based precoding (wherein the different precoders are picked from a priori fixed codebooks), a framework applicable, for example, to limited feedback settings. (For instance, in current 3GPP cellular systems, the channel feedback from the mobile station (MS) to the base station (BS) consists of an index corresponding to the preferred codebook precoder, as opposed to more detailed analog feedback of the channel.) Recent related work \cite{Omar, Alkhateeb, Omar_ASA} has considered channel sparsity based low complexity hybrid precoding for mmwave systems in different contexts (e.g., employing a "fully connected" architecture (each RF chain is connected to all the antennas), in contrast to the array-of-subarrays architecture considered here, and/or focussing on the unconstrained (i.e., non-codebook based) precoding problem. To the best of our knowledge, besides our preliminary results reported in \cite{GC14}, codebook based precoding with the array-of-subarrays architecture has not been investigated in the literature. Another related set of prior works considers beamforming optimization for indoor mmwave systems (e.g., systems based on the IEEE 802.11ad and 802.15.3c standards). In particular, reduced complexity beam search methods have been considered, e.g.,  via a two-level beam search protocol (sector level sweep, followed by a beam refinement procedure within the selected sector) \cite{IEEE_802_15_3c}, and via advanced numerical optimization techniques \cite{Li_WPAN}. All such prior work, however, has been restricted to systems with a single antenna (sub)array at the Tx and the Rx, unlike the multiple antenna (sub)array architecture considered here. \looseness-1

A straightforward approach to codebook-based hybrid precoder optimization problem would be to try all possible combinations of the different precoders. However, as we observe here, the complexity of such an approach scales exponentially with the number of subarrays employed at the Tx and the Rx, with the exponential scaling arising out of the need to assign an RF beam direction to each of the subarrays (at the Tx and the Rx). In fact, the absolute number of combinations to try out scale rapidly enough to make this approach prohibitive, even for reasonable system parameters. To overcome this bottleneck, we propose and investigate reduced complexity precoding algorithms. Our algorithms are premised on exploiting the structure inherent to the mmwave channels \cite{Smulders_channel, Rapp_channel, Rapp_channel_2, channel_ad}. In particular, mmwave channels are characterized by a sparse multipath structure, which in the spatial domain corresponds to a small number of (dominant) angles of departure (AoDs) from the Tx and a small number of (dominant) angles of arrival (AoAs) at the Rx. This implies that, most of the signal power is captured in a small number of spatial directions (at both, the Tx and the Rx), so that the RF precoder search complexity (for both, the Tx and the Rx) can be reduced, in principle, by restricting attention to a small set of spatial directions. In this work, we build on the preceding motivation, and develop reduced complexity precoding algorithms. A summary of our contributions, is as follows:\looseness-1

\subsection{Contributions}
\vspace{-3mm}
\begin{itemize}
\item We consider codebook-based hybrid precoding in mmwave systems with multiple antenna arrays. Our work identifies, for the first time in the literature, the prohibitive complexity bottleneck for this problem.
\item Exploiting the sparse multipath nature of mmwave channels, we develop a systematic approach for reduced complexity precoding. Specifically, we first develop algorithms to obtain (reduced cardinality) sets consisting of only the {\emph{dominant}} RF beams, and then perform precoder optimization by searching over all precoder combinations within these sets, thereby achieving significant complexity reduction.
\item To obtain the preceding reduced cardinality sets, we develop two complementary approaches.
\begin{itemize}
\item In the first approach, we use a novel spatial power metric, termed the effective power, to select the dominant Tx and Rx RF codebook beams. Computation of the effective power averages out the spatial fluctuation of the received signal strength across the different subarrays, providing a set of candidate directions that are suitable across \emph{all} the subarrays. A detailed analysis of the power metric is conducted for large antenna array sizes, and it is shown that, in the large system limit, the selected beams approach the channel AoA and AoD directions.
\item While the preceding approach works with instantaneous channel coefficients (thereby providing instantaneous dominant beam directions), in a complementary approach, we develop an algorithm to estimate the \emph{precise} long-term dominant (i.e., average power dominant) AoAs at the Rx. Our algorithm utilizes the phase of the measured correlations in the symbols received across the spatially separated Rx subarrays, which is intricately linked to the channel AoAs. The resulting (small number of) AoA directions are then used as the candidate RF beams at the Rx.\looseness-1
\end{itemize}
\item Simulation results (including over the IMT urban micro spatial channel model \cite{IMT_SCM}) show that the proposed methods achieve performance close to that attained with exhaustive optimization, while reducing the search complexity enormously (to realistic levels), indicating the potential feasibility of using multiple arrays in mmwave systems.
\end{itemize}

\begin{figure}[]
      \centerline {\epsfig{figure=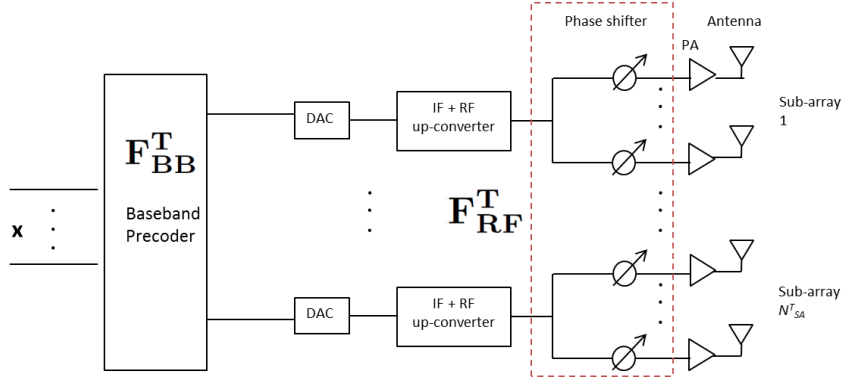, width=11.7cm, height=5.0cm}}
      \vspace{-3mm}
      \caption{Array-of-subarrays architecture for mmwave transmitter.}
       \label{fig:CSI}
       \vspace{-8mm}
 \end{figure}

\vspace{-3mm}
\section{Transceiver Architecture and Channel Model}\label{sec:Arch}
We consider an array of antennas employed at the Tx and at the Rx. The Tx (Rx) array comprises multiple subarrays (Fig. 1 depicts the Tx structure), with each subarray connected to one of the RF chains employed at the Tx (Rx). We denote by $N^T_{Ant}$ ($N^R_{Ant}$) the total number of antennas in the Tx (Rx) array. The number of Tx (Rx) subarrays is denoted by $N^{T}_{SA}$ ($N^{R}_{SA}$), while the number of antennas in each Tx (Rx) subarray is  $N^T_{Ant_{SA}}$ ($N^R_{Ant_{SA}}$). (Note that $N^T_{Ant} = N^T_{SA} \times N^T_{Ant_{SA}}$, and $N^R_{Ant} = N^R_{SA} \times N^R_{Ant_{SA}}$.) {\bf{H}} denotes the $N^R_{Ant} \times N^T_{Ant}$ (narrowband) MIMO channel,
$N_L$ is number of layers (i.e., data streams), ${\bf{F}}^{T}_{BB}$ is the ($N^T_{SA} \times N_L$) Tx BB precoder, ${\bf{F}}^{T}_{RF}$ is the ($N^T_{Ant} \times N^T_{SA}$) Tx RF precoder, and ${\bf{F}}^{R}_{RF}$ is the ($N^R_{Ant} \times N^R_{SA}$) Rx RF precoder.

Denoting by $\x$ the $N_L \times 1$ vector of transmitted symbols, and by $\y$ the $N^{R}_{Ant} \times 1$ vector of symbols received across the receiver antennas, we have
\vspace{-2mm}
\begin{equation}
\y = {\bf{H}} {\bf{F}}^{T}_{RF} {\bf{F}}^{T}_{BB} \x + {\bf{w}}  \ \ ,
\end{equation}
\vspace{-2mm}
where $\bf{w}$ is a $N^{R}_{Ant} \times 1$ noise vector, consisting of i.i.d. $CN(0, \sigma^2)$ entries. 
After RF precoding at the receiver, the $N^{R}_{SA} \times 1$ vector of symbols received across the receiver subarrays is,
\vspace{-2mm}
\begin{equation}\label{eqref:Pr_noise}
\z = {{\bf{F}}^{R}_{RF}}^* {\bf{H}} {\bf{F}}^{T}_{RF} {\bf{F}}^{T}_{BB} \x + {{\bf{F}}^{R}_{RF}}^* {\bf{w}}  \ ,
\end{equation}
\vspace{-2mm}
where $(\cdot)^*$ denotes the conjugate transpose operation.

Note that for the array-of-subarrays architecture, the RF precoder matrix (at the TX and Rx) possesses a special structure. Since each subarray is connected to only one RF chain, each column of the RF precoder matrix is zero except for a contiguous block of nonzero entries (consisting of the beamforming weights used on the corresponding subarray).

The beamforming vector in each column of the (Tx/Rx) RF precoder is normalized to have unit power, so that ${{\bf{F}}^{R}_{RF}}^* {\bf{F}}^{R}_{RF} = {\bf {I}}_{N^R_{SA}}$ (and ${{\bf{F}}^{T}_{RF}}^* {\bf{F}}^{T}_{RF} = {\bf {I}}_{N^T_{SA}}$). Consequently, the entries in the processed noise ${{\bf{F}}^{R}_{RF}}^* \bf{w}$ are still i.i.d  $CN(0, \sigma^2)$. Therefore, we can equivalently write \eqref{eqref:Pr_noise} as
\begin{equation}
\z = {{\bf{F}}^{R}_{RF}}^*{\bf{H}} {\bf{F}}^{T}_{RF} {\bf{F}}^{T}_{BB} \x + \n  \ \ ,
\end{equation}
with $\n$ consisting of i.i.d. $CN(0, \sigma^2)$.

\vspace{-2mm}
\subsection{Channel Model}\label{sec:ch_model}
\vspace{-2mm}
In contrast to the rich scattering model assumed for microwave frequencies, mmwave channels are better characterized by a limited number of scatterers. A ray-cluster based spatial channel model (each scatterer resulting in a cluster of channel rays) is typically employed \cite{Smulders_channel, Rapp_channel, Rapp_channel_2, channel_ad}. With $N^T_{Ant}$ $(N^R_{Ant})$ number of Tx (Rx) antennas, the (narrowband) channel is represented as,\looseness-1

\vspace{-6mm}
\begin{multline}\label{mmw_channel}
{\bf{H}} = \sqrt{N^T_{Ant}N^R_{Ant}}\displaystyle\sum_{c=0}^{N_{c}-1}\displaystyle\sum_{r=0}^{n_{r}-1}{G_{c,r} {\bf{a}}^R(\phi_{AoA,c,r} \ , \ \theta_{AoA,c,r})}  {{{{\bf{a}}^T}(\phi_{AoD,c,r} \ , \  \theta_{AoD, c,r})}}^* \ ,
\end{multline}
\vspace{-1mm}
where $\{G_{c,r}$, $\phi_{AoA,c,r}$, $\phi_{AoD,c,r}$, $\theta_{AoA,c,r}$, $\theta_{AoD,c,r}\}$ denote the complex gain, azimuthal AoA, AoD, and elevation AoA, AoD of ray $r$ in cluster $c$, resp., and ${\bf{a}}^R(\cdot)$ and ${\bf{a}}^T(\cdot)$ denote the array response vectors for the Rx and the Tx antenna arrays, resp. For a uniform planar array in the $yz$ plane, with $N_y$ and $N_z$ antenna elements along the $y$ and $z$ axes (with $N=N_y \times N_z$),  the array response is \cite{Balanis},
\vspace{-2mm}
\vspace{-2mm}
\begin{multline}\label{eq:array_resp}
\noindent {\bf{a}}(\phi, \theta)=\frac{1}{\sqrt{N}} [1, . . ,  e^{j k d ((n_z-1) cos (\theta) + (n_y-1) sin (\theta) sin(\phi))}, . . ,  e^{j k d ( (N_z-1)cos(\theta) + (N_y-1) sin(\theta) sin(\phi)}) ]^{'} \ ,
\end{multline}
\vspace{-2mm}
where $d$ is the inter-element spacing (along y and z dimensions) and, $1\leq n_y \leq N_y , 1 \leq n_z \leq N_z$ are element indices. While we use ${\bf{a}}^R(\cdot)$ and ${\bf{a}}^T(\cdot)$ to denote the array response vectors for the Rx and Tx arrays, we will also allude to the array response vectors corresponding to a Rx and Tx subarray, and refer to them as ${\bf{a}}^{R_{SA}}(\cdot)$ and ${\bf{a}}^{T_{SA}}(\cdot)$, respectively.

An equivalent representation of \eqref{mmw_channel} (that we allude to later), considering a single summation over all the $N_r=N_c \times n_r$ rays, is \vspace{-2mm}
\begin{equation}
\label{mmw_channel_2}
{\bf{H}} = \sqrt{N^T_{Ant}N^R_{Ant}}\displaystyle\sum_{r=0}^{N_r-1}{G_{r} {\bf{a}}^R(\phi_{AoA,r} \ , \ \theta_{AoA,r})} {{{{\bf{a}}^T}(\phi_{AoD,r} \ , \  \theta_{AoD,r})}}^* \ .
\vspace{-2mm}
\end{equation}

\vspace{-3mm}
\section{Framework}\label{sec:Framework}
\vspace{-2mm}

\noindent
\emph{Codebook based precoding:} We consider a discrete set (i.e., codebook) of possible beamforming vectors for the Tx (Rx). In particular, each of the vectors in the codebook steers the Tx (Rx) beam towards a certain ($\phi, \theta$) in the (azimuth, elevation) dimension. Notation wise, we have,
\begin{itemize}
\item ${\mathcal{C}}^T_{RF} = \{ (\phi^T_1,\theta^T_1), (\phi^T_2,\theta^T_2), \ldots, (\phi^T_{N^T_{Beams}},\theta^T_{N^T_{Beams}})\}$, is the codebook of RF beams at the Tx, where, $N^T_{Beams}$ is the number of possible Tx RF beams.
\item ${\mathcal{C}}^R_{RF} = \{ (\phi^R_1,\theta^R_1), (\phi^R_2,\theta^R_2), \ldots, (\phi^R_{N^R_{Beams}},\theta^R_{N^R_{Beams}})\}$, is the codebook of RF beams at the Rx, where, $N^R_{Beams}$ is the number of possible Rx RF beams.
\end{itemize}

Each subarray at the Tx (Rx) can pick any of the beams in ${\mathcal{C}}^T_{RF} ({\mathcal{C}}^R_{RF})$. In terms of the RF precoder matrix at the Tx (Rx), this implies that the nonzero vector in each column can be picked to be the vector corresponding to any of the beams in ${\mathcal{C}}^T_{RF} ({\mathcal{C}}^T_{RF})$. For $N^T_{SA}$ subarrays at the Tx, we therefore have ${\bf{F}}^T_{RF} \in {({\mathcal{C}}^T_{RF})}^{N^T_{SA}}$. Similarly, ${\bf{F}}^R_{RF} \in {({\mathcal{C}}^R_{RF})}^{N^R_{SA}}$.

The Tx BB precoder matrix is also picked from a specified codebook of matrices. Specifically, ${\bf{F}}^T_{BB} \in {\mathcal{C}}^T_{BB} = \{{\bf{P}}_1, {\bf{P}}_2, \ldots, {\bf{P}}_{N^T_{BB}}\}$, where, ${\mathcal{C}}^T_{BB}$ is the Tx BB precoder codebook. It consists of $N^T_{BB}$ precoding matrices, with the $i^{th}$ matrix being ${\bf{P}}_i$.

\noindent
\emph{CSI-RS transmission and channel measurements:} 
In current cellular systems (3GPP LTE), the BS transmits a reference symbol from each Tx antenna, so that the Rx antennas can sense the channel from the different Tx antennas without any interference. For mmwave systems, since the Tx and Rx subarrays (the analogues of the Tx and Rx antennas in LTE systems) can beamform in several possible directions, CSI-RS are transmitted from each subarray at the Tx, so as to enable channel measurements corresponding to different beam pair combinations at the Tx and Rx subarrays. In particular, for $N^T_{Beams}$ and $N^R_{Beams}$ number of beams at the Tx and the Rx, resp., a particular Tx subarray transmits $N^T_{Beams} \times N^R_{Beams}$ CSI-RS symbols, with the transmissions from different Tx subarrays orthogonalized in time/freqeuncy to prevent interference. 
After scanning the CSI-RS symbols transmitted by all the Tx subarrays, the Rx can acquire (estimates of) the following channel coefficients: $\{h_{i,j, b_R, b_T}\}$, where $i\in \{1, \ldots, N^R_{SA}\}$ is the Rx subarray index, $j\in \{1, \ldots, N^T_{SA}\}$ is the Tx subarray index, $b_R\in \{1, \ldots, N^R_{beams}\}$ is the beam index picked at the Rx subarray, and, $b_T\in \{1, \ldots, N^T_{beams}\}$ is the beam index picked at the Tx subarray. Using these measurements, the Rx needs to perform a joint optimization of the Tx/Rx RF precoders and the Tx BB precoder, and feed back the Tx RF/BB precoder choices to the Tx for subsequent data transmission. \looseness-1

\vspace{-2mm}
\section{Hybrid Precoder Optimization}
\vspace{-1.5mm}
For every choice of the BB and RF precoder at the Tx and the RF precoder at the Rx, we get an overall compressed channel ${\bf{H}}_{c} =  {{\bf{F}}^{R}_{RF}}^{*} {\bf{H}} {\bf{F}}^{T}_{RF} {\bf{F}}^{T}_{BB}$, so that the transmission equation (3)  becomes $\z={\bf{H}}_{c} \x  + \n$. (The entries of the matrix ${\bf{H}}_c$ can be obtained in a straightforward manner from the set of CSI-RS channel measurements $\{h_{i,j, b_R, b_T}\}$.)

Using the mutual information achieved over this channel as the optimization criterion, the optimization problem is
\vspace{-6mm}
\begin{equation}\label{eq:opt_eq}
\displaystyle \argmax_{{\bf{F}}^T_{BB} \in \mathcal{C}^T_{BB} \ {\bf{F}}^T_{RF} \in ({\mathcal{C}^T_{RF}})^{N^T_{SA}}, {\bf{F}}^R_{RF} \in ({\mathcal{C}^R_{RF}})^{N^R_{SA}}}{\log_2{\det({{\bf {I}}+\frac{1}{\sigma^2} {{\bf{H}}_{c}}^*{\bf{H}}_{c}}}}) \ \ .
\end{equation}

\vspace{-4mm}
\subsection{Complexity}
\vspace{-3mm}

A direct approach to perform the preceding optimization is to evaluate the mutual information for all possible precoder combinations. Since ${\bf{F}}^T_{RF} \in {({\mathcal{C}}^T_{RF})}^{N^T_{SA}}$, there are a total of $(N^T_{Beams})^{N^T_{SA}}$ possible choices for the Tx RF precoder. Similarly, there are a total of $(N^R_{Beams})^{N^R_{SA}}$ possible choices for the Rx RF precoder. Since the BB precoder at the Tx can be picked from amongst $N^T_{BB}$ precoder matrices, the total number of combinations to consider are
\vspace{-2mm}\begin{equation}\label{eq:Complexity}
K =  (N^R_{Beams})^{N^R_{SA}} \times (N^T_{Beams})^{N^T_{SA}} \times N^T_{BB} \  .
\vspace{-7mm}
\end{equation}
\vspace{-3mm}

The total number of combinations scales exponentially with the number of subarrays used at the Tx and the Rx. Even for reasonable values of the system parameters, this exponential scaling makes the number of combinations so large that it severely prohibits an exhaustive search. For instance, with 4 Tx subarrays and 2 Rx subarrays (akin to a $4 \times 2$ MIMO configuration in LTE), and 8 possible beams at the Tx and the Rx (typical numbers that we expect to be used for system design, based on our simulations), we get $K=8^2 \times 8^4 \times N^T_{BB} = 2^{18} \times N^T_{BB} \  ! $ Note that, in contrast, in 3GPP LTE, the Rx needs to optimize over the choice of $N^T_{BB}$ baseband precoders (only), which is already known to make the CSI feedback computation module quite resource intensive \cite{CCNC}. Further, it is important to remark that in a cellular environment, mobility of the receiver would mandate performing this optimization at regular intervals, which threatens the very applicability of hybrid precoding to millimeter wave cellular systems. 

\vspace{-5mm}
\subsection{Reduced Complexity Precoding}
\vspace{-1.5mm}

The major contribution to the high complexity of precoder selection \eqref{eq:Complexity} comes from the selection of the RF beams at the different Tx/Rx subarrays. Towards reducing this complexity, we note that, due to high propagation losses and limited environmental scattering, the mmwave channel is characterized by a small number of (dominant) paths between the Tx and the Rx. This sparse nature of the mmwave channel implies that most of the signal energy is expected to be concentrated around a small set of spatial directions, which opens up the possibility for reducing the RF beam search space, by way of restricting attention to a subset of beams that captures most of the signal energy. In the next section, we present an approach to obtain the desired Tx and Rx codebook subsets of dominant beams, based on a power metric (computed using  the strengths of the CSI-RS channel measurements) that appropriately accounts for the spatial variation of the channel energy across different subarrays while selecting the dominant beam directions. A detailed analysis of the power metric is performed in the limit of large array sizes, and it is shown that, for large systems, the resulting Tx/Rx beam directions approach the channel AoAs/AoDs. Complementary to the preceding approach, in Section \ref{sec:AoA_estimator}, we develop another method for reduced complexity precoding in mmwave systems. Specifically, we exploit the signal correlations across the Rx subarrays to develop an algorithm for estimating the precise channel AoA directions at the Rx, and use these estimated directions as the candidate beams for reception. Since the number of dominant AoA directions is expected to be small, this cuts down on the RF beam search space at the receiver. Simulation results evaluating the performance of the proposed algorithms are provided in Section \ref{sec:Results}.

\vspace{-3mm}
\section{Dominant Beam Selection for Reduced Complexity Precoding}\label{sec:Eff_power}
\vspace{-2mm}

The typical CSI-RS measurement, $h_{i,j,b_R,b_T}$, provides (a noisy estimate of)  the channel between subarray $i$ at the Rx and subarray $j$ at the Tx, when the Rx subarray is steered towards beam index $b_R$ and the Tx subarray is steered towards beam index $b_T$. We are interested in obtaining subsets of Tx and Rx beam codebooks, that capture most of the signal energy. However, note that the signal strengths are not only a function of the Tx and Rx beam indices, but also dependent on the Tx and Rx subarray indices. For instance, a \emph{strong} Rx beam direction at a particular Rx subarray may actually appear to be a \emph{weak} direction when measuring the channel at another Rx subarray. This variation in the signal strength is induced by an (Rx subarray index)-dependent phase variation in the channel corresponding to each of the channel rays, which also in fact depends on the AoA and AoD of each ray. 
To see this, for Tx subarray index $1$, Tx beam index $b_T$, and Rx beam index $b_R$, consider the channel at Rx subarray index $1$. Following the notation in Section \ref{sec:ch_model}, this can be written as (a precise proof is provided in Section \ref{sec:Asymp}, Property 2),\looseness-1
\vspace{-2mm}
\begin{eqnarray}\label{}
h_{1,1,b_R,b_T} = \sqrt{N^T_{Ant_{SA}}N^R_{Ant_{SA}}} \displaystyle\sum_{r=0}^{N_r-1}{G_r {{{\bf{a}}^{R_{SA}}}(\phi^R_{b_R}, \theta^R_{b_R})}^* {\bf{a}}^{R_{SA}}(\phi_{AoA,r}, \theta_{AoA,r})} \\ \nonumber {{{\bf{a}}^{T_{SA}}(\phi_{AoD,r}, \theta_{AoD,r})}^* {\bf{a}}^{T_{SA}}(\phi^T_{b_T}, \theta^T_{b_T})} \ + \ n\ ,
\end{eqnarray}
with  ${\bf{a}}^{R_{SA}} (\cdot, \cdot)$ and ${\bf{a}}^{T_{SA}} (\cdot, \cdot)$ denoting the response vectors corresponding to the Rx and Tx subarrays, resp.

For the same Tx subarray index and the same Tx and Rx beams, the channel, seen at Rx subarray index $i$, is (again, proved in Section \ref{sec:Asymp}, Property 2)
\vspace{-2mm}
\begin{eqnarray}\label{}
h_{i,1,b_R,b_T} = & \sqrt{N^T_{Ant_{SA}}N^R_{Ant_{SA}}} \displaystyle\sum_{r=0}^{N_r-1}{G_r {{{\bf{a}}^{R_{SA}}}(\phi^R_{b_R}, \theta^R_{b_R})}^* [{\bf{a}}^{R_{SA}}(\phi_{AoA,r}, \theta_{AoA,r}) e^{-j\gamma_{r,i}^R}]} \nonumber \\ & \ \ \ \ \ \ \ \ \ {{{{\bf{a}}^{T_{SA}}}(\phi_{AoD,r}, \theta_{AoD,r})}^* {\bf{a}}_{T_{SA}}(\phi^T_{b_T}, \theta^T_{b_T})} \ + \ n \ ,
\end{eqnarray}
where $\gamma_{r,i}^R=k (d_z^{R_i} cos(\theta_{AoA,r}) + d_y^{R_i} sin(\theta_{AoA,r})sin(\phi_{AoA,r}))$, with $d_y^{R_i}$ and $d_z^{R_i}$ denoting the distances (along the $y$ and $z$ dimensions) between Rx subarray $1$ and Rx subarray $i$.

The phase variation induced by the term $e^{j k (d_z^{R_i} cos(\theta_{AoA,r}) + d_y^{R_i} sin(\theta_{AoA,r})sin(\phi_{AoA,r}))}$, captures the relative variation in the response of the Rx subarray $i$ and response of the Rx subarray $1$, to an incoming ray at $(\phi_{AoA,r}, \theta_{AoA,r})$. Note that this phase variation depends on $(\phi_{AoA,r}, \theta_{AoA,r})$, so that for channels with multiple ($>1$) rays, each of the terms inside the summation is impacted in a different manner. Consequently, for identical Tx subarray and identical Tx and Rx beams, the magnitude of the channel coefficient measured across different Rx subarrays can vary significantly, depending on the channel AoAs, AoDs, and subarray sizes. Indeed, the phase variations induced across the spatially separated Rx  subarrays appear analogous to the phase variation induced by a non-zero Doppler frequency shift in a time-varying channel. (For a Doppler frequency shift $f_D$, the phase variation (over time) in the channel contribution of an incoming ray at an angle $\alpha$ is captured by the multiplicative term $e^{j 2 \pi f_D sin(\alpha)t}$.) Note that while we have discussed the channel strength variations across Rx subarray indices, similar arguments hold when considering the variations in the channel strength corresponding to different Tx subarrays.

\vspace{-5mm}
\subsection{Dominant Beam Selection Metric} \label{sec:BS_metric}
\vspace{-2mm}
In light of the preceding discussion, since the strength of the channel coefficients, for given Tx/Rx beams can vary across the different Tx/Rx subarrays, we obtain the Tx/Rx beam subsets (to which we wish to restrict our search space) as the sets of beams that maximize an appropriately chosen \emph{average} signal strength measure. Specifically, we use the following approach: To obtain the average signal strength for a particular Rx beam direction, we average the received signal strength across all Tx/Rx subarrays, and all Tx beams, while keeping the Rx beam direction fixed to the desired direction. Mathematically, we compute the following \emph{effective}-power metric, for each of the Rx beams {\footnote{{Note that, we directly use the (noisy) CSI-RS channel measurement $h_{i,j,l,b_T}$ for computing the power metric. This noisy measurement indeed represents the maximum likelihood (ML) estimate of the true channel coefficient. Alternate estimates, such as the minimum mean squared error estimate, could possibly be employed in lieu of this ML estimate.}}},
\vspace{-2mm}
\begin{equation}\label{eq:effpower_rx}
P^R_{eff}(l) = \frac{\displaystyle\sum_{i=1}^{N^R_{SA}}\sum_{j=1}^{N^T_{SA}}\sum_{b_T=1}^{N^T_{Beams}}|h_{i,j,l,b_T}|^2 }{(N^R_{SA} \times N^T_{SA} \times N^T_{Beams})} \ \ , \ \ l \in \{1,2, \ldots, N^R_{beams}\} .
\end{equation}

A similar procedure is performed corresponding to the Tx beams. Specifically, we compute the effective powers for the Tx beams, as,
\vspace{-2mm}
\begin{equation}\label{eq:effpower_tx}
P^T_{eff}(k) = \frac{\displaystyle\sum_{i=1}^{N^R_{SA}}\sum_{j=1}^{N^T_{SA}}\sum_{b_R=1}^{N^R_{Beams}}|h_{i,j,b_R,k}|^2 }{(N^R_{SA} \times N^T_{SA} \times N^R_{Beams})} \ \ , \ \ k \in \{1,2, \ldots, N^T_{beams}\} .
\vspace{-2mm}
\end{equation}

To reduce the search space complexity, we then pick the best $P$ Rx beams (with largest effective powers) and best $P$ Tx beams (with largest effective powers), and perform a search over only these subsets. The parameter $P$ can be considered as a tunable parameter, which can be picked in accordance with the tolerable search complexity. Crucially, though, given the sparse nature of the mmwave channel, we expect that a small value of $P$ would give performance close to that achieved with an exhaustive search over all beams.

The total number of combinations to consider, under our approach, is
\vspace{-3mm}
\begin{equation}\label{KpEqn}
K_P =  P^{N^R_{SA}} \times P^{N^T_{SA}} \times N^T_{BB}.\vspace{-2mm}
\end{equation}
For small values of P, which will be the paradigm of interest in mmwave communication, we expect $K_P < < K$ (cf. \eqref{eq:Complexity}), resulting in significant complexity savings. Note that, while the search space still scales exponentially in the number of Tx/Rx subarrays, the procedure affords us the flexibility to \emph{systematically} allocate the beams from the reduced subsets across the different subarrays in an optimal manner (given the reduced codebook subsets, we maximize the mutual information over all possible precoding combinations), as opposed to alternate heuristic allocations.\looseness-1 

\emph{Complexity of dominant beam selection}: To obtain the effective power for a particular Rx beam direction, we need to perform $O(N^R_{SA} \times N^T_{SA} \times N^T_{Beams})$
computations. Since there are $N^R_{Beams}$ beams at the Rx, the total number of computations to be performed for selecting the dominant beams at the Rx are $O(N^R_{SA} \times N^T_{SA} \times N^R_{Beams} \times N^T_{Beams})$. A similar number of computations are required to obtain the dominant beams at the Tx, so the total number of computations are of the same order. Hence, with the proposed method, the complexity of obtaining the dominant beams at the Tx and the Rx scales linearly with each of the parameters.

In the proposed method, we restrict the search space to beams that capture most of the channel energy, when averaged across the spatial domain. Note that, for finite antenna array sizes, it is not apparent if a metric other than the one proposed (average spatial power) can perform better. In particular, a closed form analytical expression for the beam directions that maximize the channel's mutual information appears intractable (channel AoA/AoD directions do not necessarily optimize the mutual information). In fact, even for a system with 1 Tx subarray and 1 Rx subarray, given a channel with multiple ($>1$) rays, obtaining a closed form analytical expression for the best Tx/Rx beam seems infeasible. 

However, as we consider larger antenna array sizes, focussing attention on directions in the vicinity of the channel AoAs/AoDs appears to be an intuitively plausible strategy (since the "beamwidth" is reduced with increasing array size). Next, we perform a detailed analysis of the proposed effective power metric, and show that, for large systems, the beams that are selected based on the effective power metric indeed approach the channel AoA/AoD directions.\looseness-1 

\vspace{-5mm}
\subsection{Large System Analysis}\label{sec:Asymp}
\vspace{-3mm}

Here, we analyze the effective power metric based beam selection procedure in the large system limit. (Each subarray consists of a large number of antenna elements, and the number of subarrays is large). Our main result in this section is stated in Lemma \ref{lemma:main}. Before proceeding to the main result, we first state and prove the following intermediate results of interest.

\begin{property}\label{prop:dot_p}
For a uniform planar array with $N_y$ and $N_z$ elements along $y$ and $z$ dimensions (with $N=N_y N_z$) and array response denoted by ${\bf{a}}(\phi, \theta)$ (as in \eqref{eq:array_resp}), we have
\vspace{-4mm}
\begin{multline}
\sqrt{N} {{\bf{a}}^*(\phi_1, \theta_1)}{{\bf{a}}(\phi_2, \theta_2)}
= \begin{cases}
\sqrt{N} &\text{, if $(\phi_1=\phi_2)$ and $(\theta_1 = \theta_2)$}\\
\sqrt{\frac{N_z}{N_y}}g_2(\phi_1, \theta_1, \phi_2, \theta_2) &\text{, if $(\phi_1 \neq \phi_2)$ and $(\theta_1 = \theta_2)$}\\
\frac{1}{\sqrt{N}} g_1(\theta_1, \theta_2) g_2(\phi_1, \theta_1, \phi_2, \theta_2) &\text{, else}\\
\end{cases}
\end{multline}
where $g_1(\theta_1, \theta_2)=\frac{1-e^{j k d N_z (cos(\theta_1)-cos(\theta_2))}}{1-e^{j k d (cos(\theta_1)-cos(\theta_2))}}$, $g_2(\phi_1, \theta_1, \phi_2, \theta_2)=\frac{1-e^{j k d N_y (sin(\theta_1)sin(\phi_1)-sin(\theta_2)sin(\phi_2))}}{1-e^{j k d (sin(\theta_1)sin(\phi_1)-sin(\theta_2)sin(\phi_2))}}$.
\end{property}

\begin{proof}
See Appendix.
\end{proof}

Note that, when ($\theta_1 \neq \theta_2$), we have the following upper bound (that is independent of $N$)
\vspace{-2mm}
\begin{equation}\label{eq:ub1}
|g_1(\theta_1, \theta_2)| < \frac{2}{|1-e^{j k d (cos(\theta_1)-cos(\theta_2))}|}  \ ,
\end{equation}
and, when, ($\theta_1 \neq \theta_2$) or ($\phi_1 \neq \phi_2)$, we have
\vspace{-2mm}
\begin{equation}\label{eq:ub2}
|g_2(\phi_1, \theta_1, \phi_2, \theta_2)| <  \frac{2}{|1-e^{j k d (sin(\theta_1)sin(\phi_1)-sin(\theta_2)sin(\phi_2))}|} \ .
\end{equation}

\begin{property}\label{prop:dot_p_exp}
The CSI-RS channel measurement $h_{i,j,b_R,b_T}$ 
can be written as
\vspace{-3mm}
\begin{multline}\label{eq:CSI-RS_meas}
h_{i,j,b_R,b_T} =
\displaystyle\sum_{r=0}^{N_r-1} e^{-j(\gamma_{r,i}^R+\gamma_{r,j}^T)}G_{r} \left[\sqrt{N^R_{Ant_{SA}}} \ {{\bf{a}}^{R_{SA}}(\phi_{AoA,r}, \theta_{AoA,r})}^* {\bf{a}}^{R_{SA}}(\phi^R_{b_R}, \theta^R_{b_R})\right] \\
\left[\sqrt{N^T_{Ant_{SA}}} \ {{\bf{a}}^{T_{SA}}(\phi_{AoD,r}, \theta_{AoD,r})}^* {\bf{a}}^{T_{SA}}(\phi^T_{b_T}, \theta^T_{b_T})\right]    + n \ .
\end{multline}
where $\gamma_{r,i}^R=k (d_z^{R_i} cos(\theta_{AoA,r}) + d_y^{R_i} sin(\theta_{AoA,r})sin(\phi_{AoA,r}))$, with $d_y^{R_i}$ and $d_z^{R_i}$ denoting the distances (along the $y$ and $z$ dimensions) between Rx subarray 1 and Rx subarray $i$. ($\gamma_{r,j}^T$ is defined in an analogous manner.)
\end{property}

\begin{proof} See Appendix.
\end{proof}
Note that, for a given ray index $r$, the products within the parentheses in \eqref{eq:CSI-RS_meas} measure the projection of the Rx (Tx) subarray's beam on to the channel AoA (AoD) direction, while the phase rotations $e^{-j\gamma_{r,i}^R}$ ($e^{-j\gamma_{r,j}^T}$) capture the variation in measurements across different Rx (Tx) subarrays. \looseness-1

Next, we proceed to the main result of this subsection, related to the large system properties of the effective power metric.
\begin{lemma}\label{lemma:main}
In the large system limit: (a) the effective power \eqref{eq:effpower_rx} measured in a channel AoA direction dominates the effective power measured in a non-AoA direction, and, (b) the effective powers for the different AoA directions are sorted in the order of the channel gains along the different AoAs. Analogous results hold for the effective power measurements \eqref{eq:effpower_tx} in the departure directions.\looseness-1
\end{lemma}
\begin{proof} See Appendix.
\end{proof}

{\emph{Implications of Lemma 1 for codebook based beam selection:}} Lemma $1$ establishes certain intuitively desirable properties of the effective power metric. As far as Rx (Tx) beam selection is concerned, the impact of Lemma 1 also depends on the actual Rx (Tx) codebooks, since these codebooks determine the directions in which the effective power is actually measured. When the antenna array sizes are large, while it is desirable to have denser RF codebooks (since the "beamwidth" gets narrower), this also increases the CSI-RS overhead, so that there are constraints to how dense the codebooks may be (thereby imposing constraints on whether the effective power would be measured in the AoA/AoD directions or not). Nonetheless, the results of Lemma 1 can be interpreted to infer that, with effective power based beam selection, for large systems, the Rx (Tx) codebook beams that are in the vicinity of the channel AoAs (AoDs) are more likely to be selected compared to other beams. (This is because the effective power can be expected to vary smoothly as a function of the beam direction). In other words, the selected beam directions \emph{approach} the channel AoAs/AoDs. On a related note, it is plausible to consider time division multiplexing of the different (Rx) (Tx) beams (over the channel coherence time) in order to realize denser codebooks without increasing the CSI-RS overhead. This enhances the likelihood that the effective power based beam selection procedure locks on to beams in close vicinity of the channel AoAs (AoDs).

\vspace{-4mm}
\section{AoA Estimation for Reduced Complexity Precoding}\label{sec:AoA_estimator}
\vspace{-2mm}
In this section, we present an algorithm to estimate the \emph{precise} channel AoA directions. In contrast to the preceding section, where the power measurements and dominant beam selection were performed based on the instantaneous channel realization, the AoA estimation algorithm developed here results in arrival directions which dominate in terms of the \emph{average} channel power (i.e., "long-term" power). These (small number of) dominant AoA direction estimates may then be employed as candidate RF beams for data reception at the receiver, thereby reducing the precoder search complexity (See Remark 1 below). Note that while AoA estimation provides complexity reduction w.r.t (with respect to) the Rx beam search, other techniques to estimate dominant beams at the Tx (such as the one in the preceding section) may be utilized in conjunction with this method. 

{\emph{Remark 1:}} In principle, we could consider mapping the estimated AoAs to the corresponding nearest beams in the Rx RF codebook. Here, in Section \ref{sec:Results}, we simulate the performance with precise steering in the AoA directions, to understand the "best" achievable performance with precise AoA estimates. Note that in a strict sense, this would require additional CSI-RS symbols (on top of the CSI-RS used to obtain channel measurements in the Rx RF codebook directions) to enable channel measurements when the receiver beamforms in the AoA directions. However, in practice, the Rx could obtain channel measurements in the (small number of) estimated AoA directions using a subset of the existing CSI-RS, skipping the measurements in a subset of the codebook beam directions (the skipped subset may vary with time, to ensure enough measurements in all directions). (This is particularly feasible here, since the AoA estimator we investigate relies only on the long term statistics of the CSI-RS measurements, which would not be impacted significantly by occasionally skipped measurements.) 

\vspace{-5mm}
\subsection{Principles of AoA Estimation}\label{AoAEstTheorySec}
\vspace{-2mm}

The problem of estimating the signal AoA using an array of Rx antennas has been studied extensively (see, e.g., \cite{KrimVibergTraditionalAoA}, \cite{RoyKailath}, \cite{Paulraj}, \cite{libertiRappaport}). Note, however, that the antenna array architecture assumed in these references is different from that employed in mmwave systems; in particular, access to the individual antenna outputs is assumed. In mmwave systems, we only have access to the beamformed samples, i.e., to a linear combination of the samples observed at the different antenna elements. Due to this, such "classical" AoA estimation methods are not directly applicable, as has been noted in the literature \cite{HuangGua_AoSAoA, PoonAoA, RamasamyAoA}. While \cite{PoonAoA, RamasamyAoA} consider AoA estimation with a single Rx array, \cite{HuangGua_AoSAoA} considers multiple Rx subarrays and provides an AoA estimation method based on correlating the observations across different subarrays, for the scenario of signal arrival from a single direction, in a static environment. In this work, we use similar principles to develop a methodology for signal AoA estimation in scenarios encountered in cellular systems, namely, where the Rx is mobile, and where the presence of multiple scatterers/reflectors in the environment leads to signal propagation and consequent signal arrival at the Rx via a number of paths.\looseness-1

In deriving the expressions in this section, we follow the channel model and notation introduced in Section \ref{sec:ch_model}. As will be shown, our AoA estimator exploits the time averaged correlation between the received signals across different Rx subarrays. Therefore, we now explicitly incorporate the time index $t$. Further, following the typical spatial channel modelling methodology (adopted, e.g., in \cite{IMT_SCM}), each ray within a cluster is characterized by a set of quasi-static quantities: average power, time delay, initial phase, Doppler frequency, and an AoA and AoD. 

Incorporating the preceding parameters, the received CSI-RS sample $h_{i,j,b_R, b_T}(t)$ between Rx subarray $i$ and Tx subarray $j$, with the Rx (Tx) subarray beamforming in the direction with index $b_R$ ($b_T$), can be written as
\vspace{-6mm}
\begin{multline} \label{eq:CSI-RS_meas_t}
h_{i,j,b_R,b_T} (t) = \\
\displaystyle\sum_{r=0}^{N_r-1} e^{-j 2 \pi f \tau_{r} }  | G_{r} | e^{j (\gamma_{r} + k.f_{D, r}.t) } e^{-j(\gamma_{r,i}^R+\gamma_{r,j}^T)}  \left[\sqrt{N^R_{Ant_{SA}}} \ {{\bf{a}}^{R_{SA}}(\phi_{AoA,r}, \theta_{AoA,r})}^* {\bf{a}}^{R_{SA}}(\phi^R_{b_R}, \theta^R_{b_R})\right] \\
\left[\sqrt{N^T_{Ant_{SA}}} \ {{\bf{a}}^{T_{SA}}(\phi_{AoD,r}, \theta_{AoD,r})}^* {\bf{a}}^{T_{SA}}(\phi^T_{b_T}, \theta^T_{b_T})\right]    + n_{i,j,b_R,b_T} \ ,
\end{multline}
where, as before,  $\gamma_{r,i}^R=k (d_z^{R_i} cos(\theta_{AoD,r}) + d_y^{R_i} sin(\theta_{AoD,r})sin(\phi_{AoD,r}))$, with $d_y^{R_i}$ and $d_z^{R_i}$ denoting the distances (along the $y$ and $z$ dimensions) between Rx subarray 1 and Rx subarray $i$. ($\gamma_{r,j}^T$ is defined in an analogous manner.) Here, $f$ denotes the narrowband center frequency, and for the ray with index $r$, $\tau_{r}$, $| G_{r} |$, $\gamma_{r}$, $f_{D, r}$ denote the delay relative to $\tau_{0}$, magnitude, initial phase (i.e., phase at $t=0$), and the Doppler frequency, respectively. (These ray parameters (and the ray AoA and AoD) are understood to be quasi-static.) Note that \eqref{eq:CSI-RS_meas_t} follows from Property 2, wherein we have now also incorporated the time variations in the channel coefficient.

To demonstrate the principles behind our AoA estimation algorithm, we consider an arrangement of one Tx subarray, and three Rx subarrays. The three Rx subarrays are indexed as \{1, 2, 3\}. With respect to Rx subarray 1, Rx subarray 2 is located at a shift of $d_{ySA}$ along the $Y$ axis, while Rx subarray 3 is located at a shift of $d_{zSA}$ along the $Z$ axis with respect to Rx subarray 1. In terms of the notation used to define $\gamma_{r,i}^R$ above, we have ($d_y^{R_2} = d_{ySA}, d_z^{R_2} = 0, d_y^{R_3} = 0, d_z^{R_3} = d_{zSA}$). We will exploit the relationship between the CSI-RS symbols received across the three Rx subarrays to obtain the AoA estimates. We consider a particular choice of the (Tx beam, Rx beam) pair (i.e., particular choice of $(b_R, b_T)$). Since the analysis we perform next holds for all such pairs, to simplify notation in the following, we drop the indices $(b_R, b_T)$. Further, the Tx subarray index is also dropped (since there is only 1 Tx subarray). Specifically, we denote by $z_i(t)$ $\{i=1,2,3\}$, the CSI-RS symbol received at Rx subarray $i$ at time $t$.

Using \eqref{eq:CSI-RS_meas_t}, we can write,
\vspace{-2mm}
\begin{equation}
z_{i}(t) = \sum_{r = 0}^{N_{r} - 1} z_{i, r}(t) + n_{i}(t) \ ,
\label{subarray0RxSampleTopLevelEqn}
\end{equation}
where $z_{i, r}(t)$ is the term inside the summation for ray $r$ in \eqref{eq:CSI-RS_meas_t}.

Considering Rx subarrays 1 and 2, with subarray 2 located at a shift of $d_{ySA}$ along the $Y$ axis w.r.t subarray 1, we have
\vspace{-4mm}
\begin{equation}
z_{2, r}(t) = e^{ j k.d_{ySA} . sin(\theta_{AoA, r}).sin(\phi_{AoA, r}) }.z_{1, r}(t) \ .
\label{subarray1RxSampleRxPhaseEqn}
\end{equation}

Now, consider $E \left\{ z_{2}  (t) . z^{*}_{1}  (t) \right\}$, with the expectation taken over time. Using \eqref{subarray0RxSampleTopLevelEqn}, \eqref{subarray1RxSampleRxPhaseEqn}, we get
\vspace*{-0.2cm}
\begin{flalign}
E \left\{ z_{2}  (t) . z^{*}_{1}  (t) \right\} = \sum_{r = 0}^{N_{r} - 1} | \left( \cdot \right) |^{2} . e^{ j k . d_{ySA} .sin(\theta_{AoA, r}).sin(\phi_{AoA, r}) } + \sum_{ p = 0}^{N_{r} - 1} \sum_{ q = 0; q \neq p}^{N_{r} - 1} E \left\{ C_{p, q} \right\},
\label{subarray10PartialCorrEqn}
\end{flalign}
where the noise samples observed at the two subarrays are assumed uncorrelated with each other, as well as with any of the signal components.
The cross term $C_{p, q}$ is given as
\vspace*{-0.2cm}
\begin{equation}
 C_{p, q}
 = e^{j [ (\gamma_{p} - \gamma_{q}) + k.(f_{D, p} - f_{D, q}).t ] } . g \left( \cdot \right),
\label{subarray10RayCrossTermEqn}
\end{equation}
where $g \left( \cdot \right)$ is independent of time. We note that the quantity $e^{j [ (\gamma_{p} - \gamma_{q}) + k.(f_{D, p} - f_{D, q}).t ] }$ is a unit-magnitude phasor with an initial phase, rotating in time; hence, we have $E \left\{ e^{j [ (\gamma_{p} - \gamma_{q}) + k.(f_{D, p} - f_{D, q}).t ] } \right\}  = 0$, due to which all the cross terms in (\ref{subarray10PartialCorrEqn}) vanish, leaving us with only the first (time-independent) term; hence we have
\vspace{-2mm}
\begin{flalign}
& E \left\{ z_{2}  (t) . z^{*}_{1}  (t) \right\} = \sum_{ r = 0}^{N_{r} - 1} | \left( \cdot \right) |^{2} . e^{ j k . d_{ySA} .sin(\theta_{AoA, r}).sin(\phi_{AoA, r}) } \ .
\label{subarray10CorrEqn}
\end{flalign}

\vspace{-0.2cm}
Proceeding in an analogous manner and considering Rx subarray 3 displaced by $d_{zSA}$ w.r.t Rx subarray 1 along the $Z$ axis, we have the counterpart expressions to (\ref{subarray1RxSampleRxPhaseEqn}, \ref{subarray10CorrEqn}) as
\vspace{-0.2cm}
\begin{equation}
z_{3, r}(t) = e^{ j k.d_{zSA} . cos(\theta_{AoA, r}) }.z_{1, r}(t) \ ,
\label{subarray2RxSampleRxPhaseEqn}
\end{equation}
\vspace{-0.8cm}
\begin{flalign}
& E \left\{ z_{3}  (t) . z^{*}_{1}  (t) \right\} = \sum_{ r = 0}^{ N_{r} - 1} | \left( \cdot \right) |^{2} . e^{ j k . d_{zSA} .cos(\theta_{AoA, r}) } \ .
\label{subarray20CorrEqn}
\end{flalign}
\vspace{-0.5cm}

Now, working under the assumption that, in \eqref{subarray10CorrEqn}, \eqref{subarray20CorrEqn}, a set of rays with a small angle spread clustered around a certain elevation and azimuth angle dominates the power distribution across the rays (essentially assuming that, \emph{once we have beamformed at the Tx and Rx subarrays}, rays within a small angular spread contribute a dominant fraction of the post-beamforming captured power), we can approximate \eqref{subarray10CorrEqn} and \eqref{subarray20CorrEqn} as the product of an amplitude-only term and a phase term, as follows\looseness-1
\vspace{-5mm}
\begin{equation}
 E \left\{ z_{2}  (t) . z^{*}_{1}  (t) \right\}
 \approx |(\cdot)|^{2} . e^{ j k . d_{ySA} . sin(\theta_{AoAEff}) . sin(\phi_{AoAEff}) } \ ,
\label{subarray10Corr1DSimplifiedEqn}
\end{equation}
\vspace*{-0.6cm}
\vspace{-2mm}
\begin{equation}
 E \left\{ z_{3}  (t) . z^{*}_{1}  (t) \right\}
 \approx |(\cdot)|^{2} . e^{ j k . d_{zSA} . cos(\theta_{AoAEff}) } \ ,
\label{subarray20Corr1DSimplifiedEqn}
\end{equation}
\vspace*{-0.2cm}
where $\theta_{AoAEff}$ and $\phi_{AoAEff}$ are the effective elevation and azimuth signal angles-of-arrival, representing the (closely clustered) arrival angles of the dominant rays. Given the expected correlations in \eqref{subarray20Corr1DSimplifiedEqn} and \eqref{subarray10Corr1DSimplifiedEqn} (which can be estimated based on received CSI-RS samples, as discussed next), we can estimate $\theta_{AoAEff}$  via \eqref{subarray20Corr1DSimplifiedEqn}, followed by the estimation of $\phi_{AoAEff}$ via \eqref{subarray10Corr1DSimplifiedEqn}. (Note that the preceding analysis holds for a particular choice of Tx-Rx RF beam pair. For different beam pairs, the closely clustered set of rays that dominate the power distribution in \eqref{subarray10CorrEqn}, \eqref{subarray20CorrEqn} is expected to be different, enabling estimation of different effective arrival directions.) As noted earlier, the estimated AoAs should be understood to be longer-term dominant AoAs. \looseness-1

The observation that the signal AoAs may be obtained from the phase of the subarray cross-correlations forms the basis for our AoA estimation procedure, described next.

\vspace{-4.5mm}
\subsection{AoA Estimation Procedure} \label{sec:aoaProcedureSec}
\vspace{-2mm}
As discussed in Section~\ref{sec:Framework}, the BS periodically transmits CSI-RS symbols from its subarrays via the RF beams from the codebook ${\mathcal{C}}^T_{RF} = \{ (\phi^T_1, \theta^T_1), \ldots, ( \phi^T_{N^T_{Beams}}, \theta^T_{N^T_{Beams}}) \}$. The MS receives these transmissions via the RF beams from the codebook ${\mathcal{C}}^R_{RF} = \{( \phi^R_1, \theta^R_1), \ldots, (\phi^R_{N^R_{Beams}}, \theta^R_{N^R_{Beams}}) \}$. Following the previous subsection, the AoA estimation procedure we describe now, considers a particular Tx subarray and three Rx subarrays, with Rx subarray 2 displaced $d_{ySA}$ along the $Y$ axis w.r.t Rx subarray 1, and Rx subarray 3 displaced $d_{zSA}$ along the $Z$ axis w.r.t Rx subarray 1. Though the description is w.r.t these subarrays, the procedure may be concurrently repeated for any such combination of three Rx and one Tx subarrays from among the Rx and Tx subarrays, with additional combining of the resulting estimates (e.g., for noise averaging). 
The steps in the AoA estimation procedure are as follows:

\begin{enumerate}
\item Step 1: Each CSI-RS transmission from the BS subarray is received via identical RF beam directions at the three MS subarrays; this ensures the validity of the assumptions in deriving ~(\ref{subarray1RxSampleRxPhaseEqn}, \ref{subarray10CorrEqn})~and~(\ref{subarray2RxSampleRxPhaseEqn}, \ref{subarray20CorrEqn}). (The Tx-Rx beam pair varies over different CSI-RS transmissions, as per a specified CSI-RS schedule.)
\item Step 2: Based on the received CSI-RS symbols, for every Tx-Rx RF beam pair $\{( \phi^T_j, \theta^T_j ), \\ ( \phi^R_k, \theta^R_k ) \}$, the MS calculates the quantities $P_{Avg}\left(j, k\right) = \frac{1}{M} \sum\limits_{l = 0}^{M-1} |z_{1}(t_l)|^2$, $C_{2, 1}\left(j, k \right) = \frac{1}{M} \sum\limits_{l = 0}^{M-1} z_{2}(t_l) \cdot z^{*}_{1}(t_l)$ and $C_{3, 1}\left(j, k \right) = \frac{1}{M} \sum\limits_{l = 0}^{M-1} z_{3}(t_l) \cdot z^{*}_{1}(t_l)$. Here, the time indices $t_l, l = 1, \ldots, M$ represent $M$ CSI-RS transmission instances, while $z_{1}(t_l)$, $z_{2}(t_l)$ and $z_{3}(t_l)$ represent the observed signals at time $t_l$ at receive subarrays 1, 2 and 3, respectively. The quantities $P_{Avg}(\cdot, \cdot)$, $C_{2, 1}(\cdot, \cdot)$ \& $C_{3, 1}(\cdot, \cdot)$ represent empirical calculations of the average received power at Rx subarray 1 (and Rx subarray 2 or 3) and the subarray cross-correlations.\looseness-1
\item Step 3: The largest $P$ values of $P_{Avg}(\cdot, \cdot)$ are identified and arranged in decreasing order; the $m^{th}$ sorted value is denoted as $P^{'}_{Avg}\left(j_m, k_m\right)$, with $j_m$ and $k_m$ denoting the corresponding Tx and Rx RF beam indices, respectively. 
\item Step 4: Motivated by (\ref{subarray20Corr1DSimplifiedEqn} , \ref{subarray10Corr1DSimplifiedEqn}), and making use of the empirically calculated subarray cross-correlations, the estimate of the elevation and azimuth AoAs for the $m^{th}$ arrival direction are then calculated as
\vspace{-2mm}
\vspace{-2mm}
\begin{equation}
\hat{\theta}_{AoA, m} = cos^{-1}\left( \frac{ \arg\left[ C_{3, 1}( j_m, k_m)\right]   }{   \frac{2 \pi}{ \lambda} d_{zSA} }  \right) \ , \
\hat{\phi}_{AoA, m} = sin^{-1}\left( \frac{ \arg\left[ C_{2, 1}( j_m, k_m)\right]   }{   \frac{2 \pi}{ \lambda} d_{ySA} sin(\hat{\theta}_{AoA, m})  }  \right) \ ,
\label{aoaEstEqn}
\end{equation}
where $\arg(\beta)$ represents the phase of the complex number $\beta$.
\begin{itemize}
\item Note 1: In practice, the calculated phases $\arg[C_{3, 1}(j_m, k_m)]$ \& $\arg[C_{2, 1}(j_m, k_m)]$ would be wrapped-around versions (to the range $[-\pi, \pi]$) of their actual values. It is crucial to unwrap these phases before using \eqref{aoaEstEqn}. This is discussed in the next paragraph.\looseness-1
\item Note 2: From (\ref{aoaEstEqn}), it is seen that only AoAs within the visible region of the receive subarrays, i.e., $\theta_{AoA} \in \left[ 0, \pi \right]$ \& $\phi_{AoA} \in \left[ -\frac{\pi}{2}, \frac{\pi}{2}\right]$, may be estimated unambiguously. This is reasonable in practice since AoAs outside this range would correspond to being "behind" the subarrays, and would be covered by other subarrays.
\end{itemize}
\end{enumerate}

{Unwrapping the phase:} Consider the relationship between the parameter $\delta$ and $\arg(e^{j\delta})$, where $\arg(e^{j\delta}) \in [-\pi, \pi]$. Assuming that $\delta \in [(2\alpha-1)\pi, (2\alpha+1)\pi]$ for some integer $\alpha$, we have $\arg(e^{j\delta}) = \delta - (2\alpha)\pi$, so that, $\delta = \arg(e^{j\delta}) +  (2\alpha)\pi$. Note that to use this relationship for obtaining $\delta$, given $\arg(e^{j\delta})$, one must know $\alpha$ to begin with. For our problem, considering \eqref{subarray20Corr1DSimplifiedEqn}, we need to obtain $k . d_{zSA} . cos(\theta_{AoAEff})$, based on $\arg(e^{j k . d_{zSA} . cos(\theta_{AoAEff}))}$ (which is estimated empirically as $\arg[C_{3, 1}(j_m, k_m)]$). Assuming that $k . d_{zSA} . cos(\theta_{AoAEff}) \in [(2\alpha-1)\pi, (2\alpha+1)\pi]$ for some integer $\alpha$, we could unwrap $\arg(k . d_{zSA} . cos(\theta_{AoAEff}))$ by adding $(2\alpha)\pi$ to it. While we do not know the $\alpha$ {\it{a priori}}, (we do not know $\theta_{AoAEff}$ {\it{a priori}}, as that is the quantity being estimated), a good choice to obtain $\alpha$ is to use the elevation reception angle of the Rx codebook beam under consideration (i.e., to use Rx RF codebook angle $\theta^R_{k_m}$)  in lieu of the unknown $\theta_{AoAEff}$. (From the discussion surrounding \eqref{subarray20Corr1DSimplifiedEqn}, we expect that when the Rx beamforms in the direction of beam index $k_m$, the estimated elevation AoA would correspond to arrival angles in the vicinity of $\theta^R_{k_m}$). While we describe unwrapping for elevation AoA estimation, once the elevation AoA is estimated, a similar procedure can be employed to unwrap the measured phase for azimuth AoA estimation.\looseness-1

{\emph{Performance evaluation:}} Due to space constraints, we skip results on AoA estimator's direct performance evaluation (we show throughput results with RF beam selection based on the estimated AoAs in Section \ref{sec:Results}). We do, however, remark that while prior work on AoA estimation in the array-of-subarrays context \cite{HuangGua_AoSAoA} considered the scenario of a single signal arrival angle, wherein the estimated AoA could be compared directly with this arrival angle, such an approach to evaluate the AoA estimator's performance will yield insufficient insight when we have a distribution of the channel energy among the signals arriving from multiple directions, as considered here. Consequently, to evaluate the performance of the AoA estimator, it is important to define and identify a set of "true" arrival directions, and then compare the estimated AoAs with these "true" AoAs, using a suitably defined metric. We considered one such metric, and verified the performance of the AoA estimator w.r.t this metric (over channel model \cite{IMT_SCM}). Details related to this study can be provided by the authors upon request.\looseness-1

{\emph{Emulating the output of the AoA estimator:}} In light of our AoA estimation procedure, we now make note of a method to emulate the AoA estimator in simulations that require AoA estimates for study of other algorithms. (E.g., we need AoA estimates to study the throughput performance of reduced complexity precoding. The following method allows us to emulate the output of the AoA estimator for this study, reducing the overall simulation run-time considerably). As mentioned in Step~2, the quantities $P_{Avg}(\cdot, \cdot)$, $C_{2, 1}(\cdot, \cdot)$ \& $C_{3, 1}(\cdot, \cdot)$, from which the AoA estimates are derived, represent empirical calculations of $E\left( |z_{1}(t) |^2 \right)$, $E \left\{ z_{2}  (t) . z^{*}_{1}  (t) \right\}$ \& $E \left\{ z_{3}  (t) . z^{*}_{1}  (t) \right\}$ respectively, and will converge to them (respectively) over a reasonably large number of CSI-RS transmissions. However, we note that for a particular channel realization, i.e., for a particular choice of the parameters $N_r$, and $\left\{ \tau_{r}, | G_{r} |, \gamma_{r}, f_{D, r}, \phi_{AoA, r}, \phi_{AoD, r}, \theta_{AoA, r}, \theta_{AoD, r}  \right\}$, $r = 0, \ldots, N_r-1$, the values of $E\left( |z_{1}(t) |^2 \right)$, $E \left\{ z_{2}  (t) . z^{*}_{1}  (t) \right\}$ \& $E \left\{ z_{3}  (t) . z^{*}_{1}  (t) \right\}$ can be calculated analytically. Hence, while running simulations, instead of the time-intensive technique of calculating $P_{Avg}(\cdot, \cdot)$, $C_{2, 1}(\cdot, \cdot)$ \& $C_{3, 1}(\cdot, \cdot)$ by time evolution of a channel realization, we can alternatively replace them (in Step~2) by the analytically calculated values that they converge to (we verified the convergence in simulations over \cite{IMT_SCM}.). The other steps in the procedure then yield the AoA estimates. In Section \ref{sec:Results}, we employ this emulation technique for our link throughput simulation studies requiring the use of AoA estimates for the purpose of RF beam selection.\looseness-1

Finally, we remark on the complexity of the described AoA estimation procedure. Since the average power in Step 2 of the procedure is computed for each beam pair, the complexity scales as $O(N^R_{Beams} \times N^T_{Beams})$. Note however, that the procedure was described using an exemplary set consisting of one Tx subarray and three Rx subarrays, although multiple such sets may be employed for improved performance. The precise complexity scaling w.r.t. the number of Tx/Rx subarrays would therefore be specific to the implementation.

\vspace{-3mm}
\section{Simulation Results}\label{sec:Results}
\vspace{-2mm}
We conducted simulations for transmission in the azimuth plane using uniform linear arrays, with inter-element spacing $d=\lambda/2$. The BS array consists of 16 antennas (divided into two independently steerable subarrays with 8 antennas each), while the MS array consists of 8 antennas (again, divided into two independently steerable subarrays with 4 antennas each). This set-up emulates a 2x2 conventional MIMO system. The BS sector spans 120 degrees around boresight, while the MS monitors a complete 180 degree region around boresight. The BS RF codebook consists of 12 beams spread uniformly in the sector, while the RF codebook at the MS consists of 8 uniformly spread beams. (Based on our simulations, for the configuration we consider, these are nominal number of beams). The baseband precoder is assumed to come from the 2x2 codebook used in the LTE standard \cite{36_211}. The channel model employed is the non-line-of-sight IMT 4G urban micro (UMi) \cite{IMT_SCM} spatial channel model, a cluster-ray based model as in \eqref{mmw_channel} (outdoor mmwave statistical channel models are not yet available; and can be expected to be sparser than the model considered here). An OFDM based implementation is considered: the OFDM numerology is based on the design specified in \cite{Farooq_WCNC} (see page 58), with FFT size of 128. The performance metric used is the mutual information (averaged over several channel realizations) attained with the different precoding schemes, considering only the center subcarrier in the OFDM system (approximating a narrowband transmission).\looseness-1

\begin{figure}[t]
\vspace{-6mm}
\hspace{-30 mm}
\includegraphics[width=210mm, height=99mm]{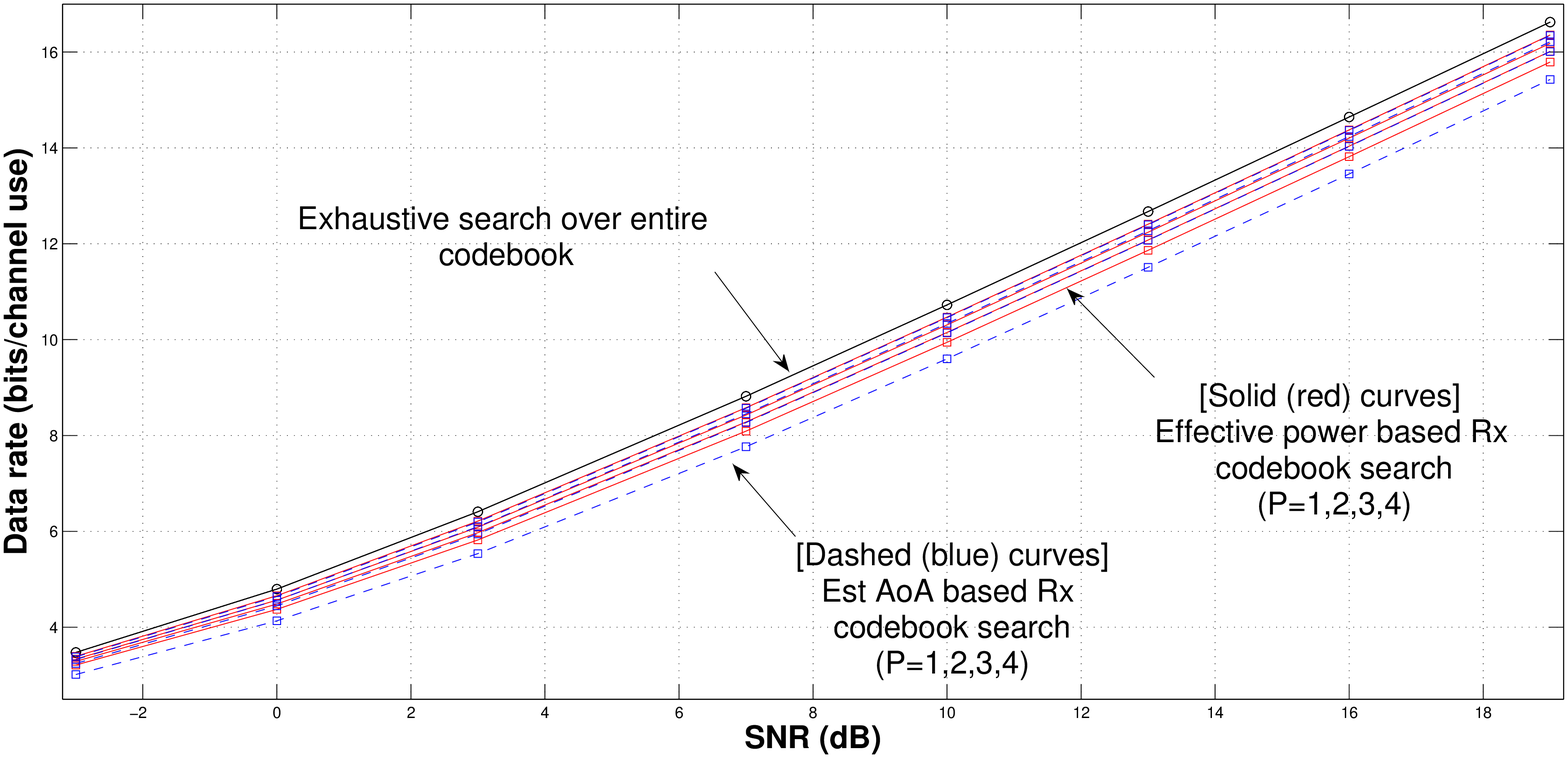}
\vspace{-13mm}
\caption{Performance of the effective power approach, AoA estimation approach, and exhaustive search.}
\label{fig:AoA_vs_heuristic}
\vspace{-8mm}
\end{figure}

In Fig. \ref{fig:AoA_vs_heuristic}, we depict the performance of the two proposed approaches: AoA estimation based complexity reduction (in MS RF beam search) and effective power based complexity reduction (applied in MS RF beam search only, to enable comparison with the AoA based approach), and the performance with an exhaustive search over the MS RF codebook. Solid (red) curves correspond to the effective power approach (for different values of the parameter \emph{P}), while dashed (blue) curves correspond to the AoA approach (for different values of the parameter \emph{P}). It is observed that both the proposed approaches provide significant complexity savings, while delivering close to optimal performance. For instance, at 10 dB SNR, with $P=1$ only, (thereby reducing the complexity by a factor of $\frac{1^2}{8^2}=\frac{1}{64}$), we can get to within 1.2 dB (1.7 dB) of the performance with exhaustive search, when using effective power approach (AoA approach). With $P=3$ (complexity reduction factor $\frac{3^2}{8^2}=\frac{9}{64}$), we are within 0.6 dB of exhaustive search, using either of the proposed approaches. This illustrates the efficacy of the proposed algorithms.\looseness-1

Comparing the performance of the effective power based approach to the estimated AoA based approach, we observed that the former approach generally performed better. (In the results shown, the effective power based approach performs significantly better when $P=1$, while the performance is very similar for other values of $P$). In hindsight, this is intuitive: Our AoA estimation algorithm attempts to lock on to directions that dominate in terms of the average (i.e., long-term) channel power, and for communication, we search over these average power dominant directions, and pick the directions with the most favorable instantaneous channel realization. On the other hand, in the effective power based approach, the dominant beam directions themselves are picked based on the instantaneous channel realization, enabling higher transmission rates. While this distinction indicates an advantage of using the effective power based method, we hasten to add that further studies are needed to arrive at definitive conclusions. In particular, the effective power metric is computed based solely on the instantaneous CSI-RS channels, and as such, to achieve this performance in practice, CSI-RS symbols corresponding to all the beam pair combinations must be transmitted at all time instants. With the AoA based approach, however, since the AoA estimator tries to lock on to the best long-term average power directions, it is plausible to reduce the CSI-RS overhead without impacting the performance of the AoA estimator significantly, e.g., by transmitting the CSI-RS for only a subset of the beam pairs at any time instant. Detailed studies in this context are an important topic for future research.\looseness-1

\begin{figure}[ht]
\vspace{-6mm}
\hspace{-20.5mm}
\includegraphics[width=210mm, height=101mm]{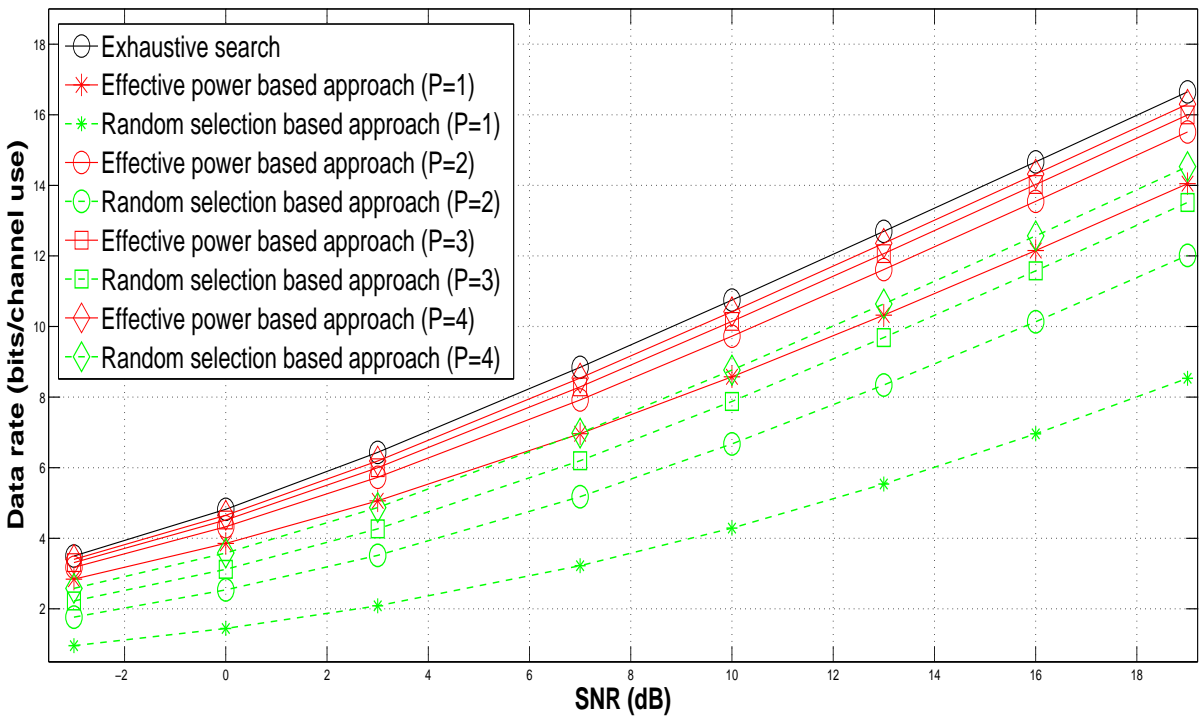}
\vspace{-13mm}
\caption{Performance of the effective power approach (used at both Tx and Rx), exhaustive search, and, random subset selection.}
\label{fig:perf_joint}
\vspace{-3mm}
\end{figure}

Next, we illustrate the performance of the effective power approach, applied to subselect beams at both the BS and the MS. To put the performance of the effective power approach into perspective, we also simulate \emph{random} beam selection, namely, we pick $P$ beams (at the BS/MS) randomly from the RF codebook (at the BS/MS), as opposed to picking beams based on the effective powers. 
Comparison with this random subset selection approach basically answers the following question: Is the proposed approach for beam selection really needed ? How much gain does it provide compared to a naive random selection of the beam subset? From the results (Fig. \ref{fig:perf_joint}), we observe that the proposed approach performs exceptionally well. For instance, at 10 dB SNR, using $P=3$, we get to within 0.9 dB of the optimal performance, and attain a complexity reduction factor of $\frac{3^2}{12^2} \times \frac{3^2}{8^2} > 100$. For the same value of $P$, the performance with random beam subset selection is 4.6 dB worse than that of exhaustive search, indicating the efficacy of the proposed approach.\looseness-1

\vspace{-5mm}
\section{Conclusions}
\vspace{-2mm}

We have investigated hybrid precoding for mmwave cellular communication with multiple antenna arrays, in the context of limited-feedback systems. While the complexity of optimal hybrid precoding is observed to be highly prohibitive, we have shown that, exploiting the sparse multipath nature of the mmwave channel, it is possible to cut down on the complexity (to realistic levels) with tolerable performance loss. This indicates that hybrid precoding, is, a feasible design option for next generation mmwave communication systems. Open technical issues include a more detailed comparison of the two complexity reduction techniques proposed here (e.g., exploring the impact of reducing the CSI-RS overhead on performance), and development of approaches that could provide further complexity reduction. Investigation of the proposed algorithms in the context of wideband systems, as well as multi-user MIMO systems, is crucial. In either system, exploiting the sparse nature of the mmwave channel to restrict attention to a set of dominant beam directions (obtained, e.g., after averaging the effective power across the subcarriers in a wideband OFDM system) can be expected to cut down the precoder optimization complexity.\looseness-1

\vspace{-5mm}
\section*{Acknowledgement}

The authors thank their colleagues from the Wireless Communications Laboratory, Samsung Research America-Dallas, and from the Communications Research Team, Samsung Electronics DMC R\&D Center, Suwon, South Korea for their support and feedback during the course of this research. In particular, they thank Dr. Taeyoung Kim for sharing the software implementation code for the UMi SCM channel.

\vspace{-2mm}

\vspace{-4mm}
\section*{Appendix}
\vspace{-3mm}
\noindent{\it{\bf{Proof of Property 1:}}}
From \eqref{eq:array_resp}, we have
\vspace{-2mm}
\begin{eqnarray}
{\bf{a}}^*(\phi_1, \theta_1){\bf{a}}(\phi_2, \theta_2) = \frac{1}{N}\displaystyle\sum_{n_z=1}^{N_z}\sum_{n_y=1}^{N_y} e^{j k d (n_z-1) (cos(\theta_1)-cos(\theta_2))} e^{j k d (n_y-1) (sin(\theta_1)sin(\phi_1)-sin(\theta_2)sin(\phi_2))} \\
= \frac{1}{N} \left[\displaystyle\sum_{n_z=1}^{N_z}{e^{j k d (n_z-1) (cos(\theta_1)-cos(\theta_2))}}\right] \left[\displaystyle\sum_{n_y=1}^{N_y}{e^{j k d (n_y-1) (sin(\theta_1)sin(\phi_1)-sin(\theta_2)sin(\phi_2))}}\right] \nonumber \ .
\end{eqnarray}
\vspace{-1mm}
If $(\phi_1=\phi_2)$ and $(\theta_1 = \theta_2)$, the right hand side (RHS) equates to $\frac{1}{N}[N_z] [N_y] = \frac{1}{N} N = 1$. If $(\phi_1 \neq \phi_2)$ and $(\theta_1 = \theta_2)$, the RHS equates to $\frac{1}{N}[N_z] [\frac{1-e^{j k d N_y (sin(\theta_1)sin(\phi_1)-sin(\theta_2)sin(\phi_2))}}{1-e^{j k d (sin(\theta_1)sin(\phi_1)-sin(\theta_2)sin(\phi_2))}}]  = \frac{1}{N_y} g_2(\phi_1, \theta_1, \phi_2, \theta_2)$. For the other two possibilities ($(\theta_1 \neq \theta_2)$ and [$(\phi_1=\phi_2)$ or $(\phi_1 \neq \phi_2)$]), the RHS equates to  $\frac{1}{N}[\frac{1-e^{j k d N_z (cos(\theta_1)-cos(\theta_2))}}{1-e^{j k d (cos(\theta_1)-cos(\theta_2))}}] [\frac{1-e^{j k d N_y (sin(\theta_1)sin(\phi_1)-sin(\theta_2)sin(\phi_2))}}{1-e^{j k d (sin(\theta_1)sin(\phi_1)-sin(\theta_2)sin(\phi_2))}}]  = \frac{1}{N} g_1(\theta_1, \theta_2) g_2(\phi_1, \theta_1, \\ \phi_2, \theta_2)$, which completes the proof.

\noindent{\it{\bf{Proof of Property 2:}}} The measurement $h_{i,j,b_R,b_T}$ is the received signal at Rx subarray $i$, when CSI-RS is transmitted from Tx subarray $j$, with the Rx (Tx) subarray beamforming in the direction with index $b_R$ ($b_T$). Specifically,
\vspace{-4mm}
\begin{equation} \label{eq:measurement}
h_{i,j,b_R,b_T} = [{{\bf{F}}^R_{RF}}^* {\bf{H}} {{\bf{F}}^T_{RF}} \ {\bf{e_j}}](i) \ ,
\end{equation}
\vspace{-3mm}
where, ${\bf{e_j}} = [0 \ 0 \ldots 1 \ldots 0 \ 0]'$ (with the $1$ occurring at location $j$), and $[\cdot](i)$ indicates the $i^{th}$  entry of the (row) vector within the square parenthesis. Further, the beamforming vector employed in row $i$ of ${{\bf{F}}^R_{RF}}^*$ is ${{{{\bf{a}}^{R_{SA}}}(\phi^R_{b_R}, \theta^R_{b_R})}}^*$, and the beamforming vector employed in column $j$ of ${{\bf{F}}^T_{RF}}$ is ${{{{\bf{a}}^{T_{SA}}}(\phi^T_{b_T}, \theta^T_{b_T})}}$, with ${\bf{a}}^{R_{SA}} (\cdot, \cdot)$ and ${\bf{a}}^{T_{SA}} (\cdot, \cdot)$ denoting the array response vectors corresponding to the Rx and Tx subarrays, resp.
We can write \eqref{eq:measurement} as,
\vspace{-3mm}
\begin{align}
h_{i,j,b_R,b_T} &= {{\bf{F}}^R_{RF}}^*(i,:) {\bf{H}} {{\bf{F}}^T_{RF}}(:,j) \label{eq:measurement2a} \\
&= \sqrt{N^T_{Ant} N^R_{Ant}}\displaystyle\sum_{r=0}^{N_r-1}{G_{r} {{\bf{F}}^R_{RF}}^*(i,:) {{\bf{a}}^R(\phi_{AoA,r} \ , \ \theta_{AoA,r})} {{\bf{a}}^T(\phi_{AoD,r} \ , \  \theta_{AoD,r})}}^* {\bf{F}}^T_{RF}(:,j) \label{eq:measurement2b} \ ,
\end{align}
\vspace{-2mm}
where ${{\bf{F}}^R_{RF}}^*(i,:)$ denotes row $i$ of ${{\bf{F}}^R_{RF}}^*$ and ${\bf{F}}^T_{RF}(:,j)$ denotes column $j$ of ${\bf{F}}^T_{RF}$.
Focussing on the product $ {{\bf{a}}^T(\phi_{AoD,r} \ , \  \theta_{AoD,r})}^* {\bf{F}}^T_{RF}(:,j)$, we have
\vspace{-3mm}
\begin{eqnarray}\label{eq:prod}
{{\bf{a}}^T(\phi_{AoD,r} \ , \  \theta_{AoD,r})}^* {\bf{F}}^T_{RF}(:,j) = {{\bf{a}}^T(\phi_{AoD,r} \ , \  \theta_{AoD,r})}^* [0 \ldots 0 \ \  {{{\bf{a}}^{T_{SA}}}(\phi^T_{b_T}, \theta^T_{b_T})}^{'} \ \  0 \ldots 0 ]^{'} \ .
\end{eqnarray}
\vspace{-2mm}
Denote by ${\bf{\tilde{a}}}^T$ the length $N^T_{Ant_{SA}}$ subset of ${{\bf{a}}^T(\phi_{AoD,r} \ , \  \theta_{AoD,r})}^*$, that actually ends up being multiplied with ${{{{\bf{a}}^{T_{SA}}}(\phi^T_{b_T}, \theta^T_{b_T})}}$ in \eqref{eq:prod}. Then, from the definition of the array response \eqref{eq:array_resp}, we can express ${\bf{\tilde{a}}}^T$ as
\vspace{-3mm}
\begin{equation}
{\bf{\tilde{a}}}^T=\frac{1}{\sqrt{N^T_{Ant}}} e^{-j k (d_z^{T_j} cos(\theta_{AoD,r}) + d_y^{T_j} sin(\theta_{AoD,r})sin(\phi_{AoD,r}))} \ \left[\sqrt{N^T_{Ant_{SA}}} \ {{\bf{a}}^{T_{SA}}(\phi_{AoD,r}, \theta_{AoD,r})}^{*}\right]
\end{equation}
\vspace{-2mm}
\noindent where $d_y^{T_j}$ and $d_z^{T_j}$ denote the distances (along the $y$ and $z$ dimensions) between Tx subarray $1$ and Tx subarray $j$. Denoting $\gamma_{r,j}^T=k (d_z^{T_j} cos(\theta_{AoD,r}) + d_y^{T_j} sin(\theta_{AoD,r})sin(\phi_{AoD,r}))$, we can thus rewrite \eqref{eq:prod} as,
\vspace{-3mm}
\begin{eqnarray}\label{eq:Txproj}
{{\bf{a}}^T(\phi_{AoD,r} \ , \  \theta_{AoD,r})}^* {\bf{F}}^T_{RF}(:,j) = \frac{e^{-j\gamma_r^T}}{\sqrt{N^T_{Ant}}}\left[\sqrt{N^T_{Ant_{SA}}} \ {{\bf{a}}^{T_{SA}}(\phi_{AoD,r}, \theta_{AoD,r})}^* {\bf{a}}^{T_{SA}}(\phi^T_{b_T}, \theta^T_{b_T})\right] \ .
\end{eqnarray}
\vspace{-2mm}
\noindent In a similar manner, we can obtain
\vspace{-2mm}
\begin{eqnarray}\label{eq:Rxproj}
{{\bf{F}}^R_{RF}}^*(i,:) {{\bf{a}}^R(\phi_{AoA,r} \ , \  \theta_{AoA,r})} = \frac{e^{-j\gamma_r^R}}{\sqrt{N^R_{Ant}}}\left[\sqrt{N^R_{Ant_{SA}}} \ {{\bf{a}}^{R_{SA}}(\phi_{AoA,r}, \theta_{AoA,r})}^* {\bf{a}}^{R_{SA}}(\phi^R_{b_R}, \theta^R_{b_R})\right] \ ,
\end{eqnarray}
\vspace{-2mm}
\noindent where $\gamma_{r,i}^R$ is defined for the Rx in an analogous manner as the preceding definition of $\gamma_{r,j}^T$ for the Tx. Substituting \eqref{eq:Txproj} and \eqref{eq:Rxproj} in \eqref{eq:measurement2b} completes the proof.

\noindent{\it{\bf{Proof of Lemma 1:}}} Since we are interested in studying the asymptotics (performance under large number of antenna elements, along the $y$ and $z$ dimensions, for each Rx/Tx subarray), w.l.o.g., in the following, we use $N_o$ to denote the number of antennas along either dimension (for both, the Tx and the Rx), and $N$ to denote the total number of antennas in each subarray (for both, the Tx, and the Rx). ($N=N_o^2$.) Referring to \eqref{eq:effpower_rx}, consider the effective power measurement for the receiver codebook beam index $l$. This beam index corresponds to an azimuth elevation pair $(\phi^R_l, \theta^R_l)$. 
Denoting, as before, the set of channel AoAs as $\{(\phi_{AoA,r}, \theta_{AoA,r}) \ ; \ 0 \leq r \leq N_r-1\}$, we consider three possible cases: (1) Rx beam is not steered towards any of the channel AoAs, in the elevation dimension, (2) Rx beam is steered towards a channel AoA in the elevation dimension, but away from the corresponding AoA in the azimuth plane, and (3) Rx beam is steered towards a channel AoA in both, elevation and azimuth. We will show that the effective power measured in Case (3) dominates the effective power measured in Cases (1) and (2).
\newline
{\emph{Case 1:}} $(\theta^R_l \neq \theta_{AoA,r}), \forall r$. In this case, using Property 1 and 2, we get
\vspace{-2mm}
\begin{eqnarray}
h_{i,j,l,b_T} =
\displaystyle\sum_{r=0}^{N_r-1} e^{-j(\gamma_{r,i}^R+\gamma_{r,j}^T)}G_{r} \left[\frac{1}{\sqrt{N}} \ g_1(\theta_{AoA,r}, \theta^R_l) \ g_2(\phi_{AoA,r}, \theta_{AoA,r}, \phi^R_l, \theta^R_l)\right] \nonumber \\
\left[\sqrt{N} \ {{\bf{a}}^{T_{SA}}(\phi_{AoD,r}, \theta_{AoD,r})}^* {\bf{a}}^{T_{SA}}(\phi^T_{b_T}, \theta^T_{b_T})\right]    + n_{i,j,l,b_T} \ .
\end{eqnarray}
(We explicitly incorporated the indices $(i,j,l,b_T)$ in the noise term as well). The product within the second parenthesis depends on the relationship between $(\phi^T_{b_T}, \theta^T_{b_T})$ and $(\phi_{AoD,r}, \theta_{AoD,r})$. Accordingly, we have the following sub-cases
\newline{Case 1a:} $(\theta^T_{b_T} \neq \theta_{AoD,r}), \forall r$. This implies that
\vspace{-2mm}
\begin{eqnarray}
h_{i,j,l,b_T} =
\displaystyle\sum_{r=0}^{N_r-1} e^{-j(\gamma_{r,i}^R+\gamma_{r,j}^T)}G_{r} \left[\frac{1}{\sqrt{N}} \ g_1(\theta_{AoA,r}, \theta^R_l) \ g_2(\phi_{AoA,r}, \theta_{AoA,r}, \phi^R_l, \theta^R_l)\right] \nonumber \\
\left[\frac{1}{\sqrt{N}} \ g_1(\theta_{AoD,r}, \theta^T_{b_T}) \ g_2(\phi_{AoD,r}, \theta_{AoD,r}, \phi^T_{b_T}, \theta^T_{b_T})\right]  + n_{i,j,l,b_T} \ .
\end{eqnarray}
Since $g_1$ and $g_2$ are bounded functions, and since the number of terms inside the summation is also bounded ($N_r$ terms), for large $N$, we get
$h_{i,j,l,b_T} \approx n_{i,j,l,b_T}$, so that $|h_{i,j,l,b_T}|^2 \approx |n_{i,j,l,b_T}|^2$.
\newline{Case 1b:} $(\theta^T_{b_T} = \theta_{AoD,k})$ for some $k$, and $(\phi^T_{b_T} \neq \phi_{AoD,k})$. Using Property 1, we get
\vspace{-2mm}
\begin{eqnarray}
h_{i,j,l,b_T} = e^{-j(\gamma_{k,i}^R+\gamma_{k,j}^T)}G_{k} \left[\frac{1}{\sqrt{N}} \ g_1(\theta_{AoA,k}, \theta^R_l) \ g_2(\phi_{AoA,k}, \theta_{AoA,k}, \phi^R_l, \theta^R_l)\right] \\
\left[\sqrt{\frac{N_o}{N_o}}g_2(\phi_{AoD,k}, \theta_{AoD,k}, \phi^T_{b_T}, \theta^T_{b_T})\right] \nonumber \\ +
\displaystyle\sum_{r=0 \ ; r \neq k}^{N_r-1} e^{-j(\gamma_{r,i}^R+\gamma_{r,j}^T)}G_{r} \left[\frac{1}{\sqrt{N}} \ g_1(\theta_{AoA,r}, \theta^R_l) \ g_2(\phi_{AoA,r}, \theta_{AoA,r}, \phi^R_l, \theta^R_l)\right] \nonumber \\
\left[\frac{1}{\sqrt{N}} \ g_1(\theta_{AoD,r}, \theta^T_{b_T}) \ g_2(\phi_{AoD,r}, \theta_{AoD,r}, \phi^T_{b_T}, \theta^T_{b_T})\right]  + n_{i,j,l,b_T} \nonumber \ .
\end{eqnarray}
Again, for large $N$, we get $h_{i,j,l,b_T} \approx n_{i,j,l,b_T}$, so that $|h_{i,j,l,b_T}|^2 \approx |n_{i,j,l,b_T}|^2$.
\newline{Case 1c:} $(\theta^T_{b_T} = \theta_{AoD,k})$ and $(\phi^T_{b_T} = \phi_{AoD,k})$ for some $k$. Using Property 1, we get
\vspace{-2mm}
\begin{eqnarray}
h_{i,j,l,b_T} = e^{-j(\gamma_{k,i}^R+\gamma_{k,j}^T)}G_{k} \left[\frac{1}{\sqrt{N}} \ g_1(\theta_{AoA,k}, \theta^R_l) \ g_2(\phi_{AoA,k}, \theta_{AoA,k}, \phi^R_l, \theta^R_l)\right] \left[\sqrt{N}\right]\\ +
\displaystyle\sum_{r=0 \ ; r \neq k}^{N_r-1} e^{-j(\gamma_{r,i}^R+\gamma_{r,j}^T)}G_{r} \left[\frac{1}{\sqrt{N}} \ g_1(\theta_{AoA,r}, \theta^R_l) \ g_2(\phi_{AoA,r}, \theta_{AoA,r}, \phi^R_l, \theta^R_l)\right] \nonumber \\
\left[\frac{1}{\sqrt{N}} \ g_1(\theta_{AoD,r}, \theta^T_{b_T}) \ g_2(\phi_{AoD,r}, \theta_{AoD,r}, \phi^T_{b_T}, \theta^T_{b_T})\right]  + n_{i,j,l,b_T} \nonumber \ .
\end{eqnarray}
For large $N$, we get $h_{i,j,l,b_T} \approx e^{-j(\gamma_{k,i}^R+\gamma_{k,j}^T)}G_{k} \left[g_1(\theta_{AoA,k}, \theta^R_l) \ g_2(\phi_{AoA,k}, \theta_{AoA,k}, \phi^R_l, \theta^R_l)\right] + n_{i,j,l,b_T}$. Now, $|h_{i,j,l,b_T}|^2= |e^{j(\gamma_{k,i}^R+\gamma_{k,j}^T)}h_{i,j,l,b_T}|^2 = |G_{k} g_1(\theta_{AoA,k}, \theta^R_l) \ g_2(\phi_{AoA,k}, \theta_{AoA,k}, \phi^R_l, \theta^R_l) + e^{j(\gamma_{k,i}^R+\gamma_{k,j}^T)} \\ n_{i,j,l,b_T}|^2$. Since $e^{j(\gamma_{k,i}^R+\gamma_{k,j}^T)}$ is a constant phase rotation, the random variable $e^{j(\gamma_{k,i}^R+\gamma_{k,j}^T)} n_{i,j,l,b_T}$ is still i.i.d. $CN(0, \sigma^2)$, so that we can write $|h_{i,j,l,b_T}|^2 = |G_{k} g_1(\theta_{AoA,k}, \theta^R_l) \ g_2(\phi_{AoA,k}, \theta_{AoA,k}, \phi^R_l, \theta^R_l) + n_{i,j,l,b_T}|^2$. (For ease of notation, we continue to use $n$ to denote the rotated noise, rather than introducing a new random variable). Expanding the square, we get, $|h_{i,j,l,b_T}|^2 = |G_{k}|^2 \ |g_1(\theta_{AoA,k}, \theta^R_l) \\ \ g_2(\phi_{AoA,k},  \theta_{AoA,k}, \phi^R_l, \theta^R_l)|^2 + |n_{i,j,l,b_T}|^2 + 2Re(G_{k} \  n_{i,j,l,b_T}^* \left[g_1(\theta_{AoA,k}, \theta^R_l) \ g_2(\phi_{AoA,k}, \theta_{AoA,k}, \phi^R_l, \theta^R_l)\right])$. Note that the index $k$ here is a function of the beam index $b_T$ under consideration. We will make this dependence explicit in the following.

We denote the set of Tx codebook beams that satisfy the constraints of cases 1a, 1b, 1c, as $\mathcal{C}_a^T$, $\mathcal{C}_b^T$, $\mathcal{C}_c^T$, resp. Next, towards computation of the effective power $P^R_{eff}(l)$ \eqref{eq:effpower_rx}, we define the partial sum,
\vspace{-2mm}
\begin{align}
S_{i,j,l} &=\frac{1}{{N^T_{Beams}}} \displaystyle\sum_{b_T=1}^{N^T_{Beams}}|h_{i,j,l,b_T}|^2 = \frac{1}{{N^T_{Beams}}} \left[\displaystyle\sum_{b_T \in \mathcal{C}_a^T}|h_{i,j,l,b_T}|^2 + \sum_{b_T \in \mathcal{C}_b^T}|h_{i,j,l,b_T}|^2 + \sum_{b_T \in \mathcal{C}_c^T}|h_{i,j,l,b_T}|^2\right] \\ \nonumber
&=\frac{1}{{N^T_{Beams}}} \left[\displaystyle\sum_{b_T=1}^{N^T_{Beams}} |n_{i,j,l,b_T}|^2\right] \\ \nonumber
&+\frac{1}{{N^T_{Beams}}} \left[\displaystyle\sum_{b_T \in \mathcal{C}_c^T}  {|G_{k(b_T)}|^2 \ |g_1(\theta_{AoA,k(b_T)}, \theta^R_l) \ g_2(\phi_{AoA,k(b_T)}, \theta_{AoA,k(b_T)}, \phi^R_l, \theta^R_l)|^2}\right] \\ \nonumber
&+ \frac{1}{{N^T_{Beams}}} \left[\displaystyle\sum_{b_T \in \mathcal{C}_c^T}  {2Re(G_{k(b_T)} \  n_{i,j,l,b_T}^* g_1(\theta_{AoA,k(b_T)}, \theta^R_l) \ g_2(\phi_{AoA,k(b_T)}, \theta_{AoA,k(b_T)}, \phi^R_l, \theta^R_l))} \right]
\end{align}

Averaging $S_{i,j,l}$ over the Rx and Tx subarrays (i.e., over $i$ and $j$), for large number of subarrays, the first term in the preceding summation is approximately $\sigma^2$, the second term is unaffected (independent of $i$ and $j$), while the third term is approximately zero (since $E(n^*G_k)=0, \ \forall k$). Hence, in the large system limit, we get
\vspace{-2mm}
\begin{equation}\label{eq:P_R_Eff1}
P^R_{eff}(l) \approx \sigma^2 + \frac{1}{{N^T_{Beams}}} \left[\displaystyle\sum_{b_T \in \mathcal{C}_c^T}  {|G_{k(b_T)}|^2 \ |g_1(\theta_{AoA,k(b_T)}, \theta^R_l) \ g_2(\phi_{AoA,k(b_T)}, \theta_{AoA,k(b_T)}, \phi^R_l, \theta^R_l)|^2}\right].
\end{equation}

This completes the analysis for Case 1. Since the analysis for the next 2 cases (and their subcases) is similar, we will skip some of the details in the following.
\newline
\noindent
{\emph{Case 2:}} $(\theta^R_l = \theta_{AoA,m})$ for some $m$, and $(\phi^R_l \neq \phi_{AoA,m})$. Using Property 1 and 2, we get
\vspace{-2mm}
\begin{align}
h_{i,j,l,b_T} &= e^{-j(\gamma_{m,i}^R+\gamma_{m,j}^T)}G_{m} \left[g_2(\phi_{AoA,m}, \theta_{AoA,m}, \phi^R_l, \theta_{AoA,m})\right] \\ &\left[\sqrt{N} \ {{\bf{a}}^{T_{SA}}(\phi_{AoD,m}, \theta_{AoD,m})}^* {\bf{a}}^{T_{SA}}(\phi^T_{b_T}, \theta^T_{b_T})\right]  \nonumber \\
&+ \displaystyle\sum_{r=0 \ ; r \neq m}^{N_r-1} e^{-j(\gamma_{r,i}^R+\gamma_{r,j}^T)}G_{r} \left[\frac{1}{\sqrt{N}} \ g_1(\theta_{AoA,r}, \theta^R_l) \ g_2(\phi_{AoA,r}, \theta_{AoA,r}, \phi^R_l, \theta^R_l)\right] \nonumber \\
&\left[\sqrt{N} \ {{\bf{a}}^{T_{SA}}(\phi_{AoD,r}, \theta_{AoD,r})}^* {\bf{a}}^{T_{SA}}(\phi^T_{b_T}, \theta^T_{b_T})\right]    + n_{i,j,l,b_T}  \nonumber \ .
\end{align}

\noindent
{Case 2a:} $(\theta^T_{b_T} \neq \theta_{AoD,r}), \forall r$.
For large $N$, we get $h_{i,j,l,b_T} \approx n_{i,j,l,b_T}$, so that $|h_{i,j,l,b_T}|^2 \approx |n_{i,j,l,b_T}|^2$.

\noindent{Case 2(b,1):} $(\theta^T_{b_T} = \theta_{AoD,k})$ for some $k$, with $k\neq m$, and $(\phi^T_{b_T} \neq \phi_{AoD,k})$.
For large $N$, we get $h_{i,j,l,b_T} \approx n_{i,j,l,b_T}$, so that $|h_{i,j,l,b_T}|^2 \approx |n_{i,j,l,b_T}|^2$.

\noindent{Case 2(b,2):} $(\theta^T_{b_T} = \theta_{AoD,m})$ and $(\phi^T_{b_T} \neq \phi_{AoD,m})$.
For large $N$, we get
\vspace{-2mm}
\begin{eqnarray}
h_{i,j,l,b_T} \approx e^{-j(\gamma_{m,i}^R+\gamma_{m,j}^T)}G_{m} \left[g_2(\phi_{AoA,m}, \theta_{AoA,m}, \phi^R_l, \theta_{AoA,m})\right] \\
\left[g_2(\phi_{AoD,m}, \theta_{AoD,m}, \phi^T_{b_T}, \theta_{AoD,m})\right] + n_{i,j,l,b_T}  \ . \nonumber
\end{eqnarray}

\noindent{Case 2(c,1):} $(\theta^T_{b_T} = \theta_{AoD,k})$ and $(\phi^T_{b_T} = \phi_{AoD,k})$, with $k \neq m$.
For large $N$, we get
\vspace{-2mm}
\begin{equation}
h_{i,j,l,b_T} \approx e^{-j(\gamma_{k,i}^R+\gamma_{k,j}^T)}G_{k} \ g_1(\theta_{AoA,k}, \theta_{AoA,m}) \ g_2(\phi_{AoA,k}, \theta_{AoA,k}, \phi^R_l, \theta_{AoA,m}) + n_{i,j,l,b_T}  \ .
\end{equation}

\noindent{Case 2(c,2):} $(\theta^T_{b_T} = \theta_{AoD,m})$ and $(\phi^T_{b_T} = \phi_{AoD,m})$.
For large $N$, we get
\vspace{-2mm}
\begin{eqnarray}
h_{i,j,l,b_T} \approx e^{-j(\gamma_{m,i}^R+\gamma_{m,j}^T)}G_{m} \left[g_2(\phi_{AoA,m}, \theta_{AoA,m}, \phi^R_l, \theta_{AoA,m})\right] \left[\sqrt{N}\right]
+ n_{i,j,l,b_T}  \ .
\end{eqnarray}

Following the approach in Case 1, we can write the partial sum $S_{i,j,l}$ as an addition of five summations (corresponding to the five subcases considered here).  While all five subcases contribute noise terms, non-noise (i.e., signal) terms are contributed by Cases 2(b,2), 2(c,1) and 2(c,2). Note that the signal term contributed by Case 2(c,2) to $S_{i,j,l}$ scales with $N$, while the contributions from the other two cases are bounded and independent of $N$, so that the contribution from the signal term contributed by Case 2(c,2) dominates for large $N$. Following the same principles described in Case 1, we therefore obtain (assuming that there exists a Tx codebook beam that satisfies the condition for Case 2(c,2); cf. Remark 2, page 29)
\vspace{-2mm}
\begin{equation}\label{eq:P_R_Eff2}
P^R_{eff}(l) \approx \sigma^2 + N |G_{m}|^2 |g_2(\phi_{AoA,m}, \theta_{AoA,m}, \phi^R_l, \theta_{AoA,m})|^2.
\end{equation}

\noindent {\emph{Case 3:}} $(\theta^R_l = \theta_{AoA,q})$ and $(\phi^R_l = \phi_{AoA,q})$ for some $q$. Using Property 1 and 2, we get
\vspace{-2mm}
\begin{eqnarray}
h_{i,j,l,b_T} = {e^{-j(\gamma_{q,i}^R+\gamma_{q,j}^T)}G_{q} \left[\sqrt{N}\right] \left[\sqrt{N} \ {{\bf{a}}^{T_{SA}}(\phi_{AoD,q}, \theta_{AoD,q})}^* {\bf{a}}^{T_{SA}}(\phi^T_{b_T}, \theta^T_{b_T})\right]}  \\
+ \displaystyle\sum_{r=0 \ ; r \neq q}^{N_r-1} e^{-j(\gamma_{r,i}^R+\gamma_{r,j}^T)}G_{r} \left[\frac{1}{\sqrt{N}} \ g_1(\theta_{AoA,r}, \theta^R_l) \ g_2(\phi_{AoA,r}, \theta_{AoA,r}, \phi^R_l, \theta^R_l)\right] \nonumber \\
\left[\sqrt{N} \ {{\bf{a}}^{T_{SA}}(\phi_{AoD,r}, \theta_{AoD,r})}^* {\bf{a}}^{T_{SA}}(\phi^T_{b_T}, \theta^T_{b_T})\right]    + n_{i,j,l,b_T}  \nonumber \ .
\end{eqnarray}

\noindent{Case 3a:} $(\theta^T_{b_T} \neq \theta_{AoD,r}), \forall r$.
For large $N$, we get
\vspace{-2mm}
\begin{eqnarray}
h_{i,j,l,b_T} \approx {e^{-j(\gamma_{q,i}^R+\gamma_{q,j}^T)} \ G_{q} \ g_1(\theta_{AoD,q}, \theta^T_{b_T}) \ g_2(\phi_{AoD,q}, \theta_{AoD,q}, \phi^T_{b_T}, \theta^T_{b_T})}  + n_{i,j,l,b_T} \ .
\end{eqnarray}

\noindent{Case 3(b,1):} $(\theta^T_{b_T} = \theta_{AoD,k})$ for some $k$, with $k\neq q$, and $(\phi^T_{b_T} \neq \phi_{AoD,k})$.
For large $N$, we get
\vspace{-2mm}
\begin{eqnarray}
h_{i,j,l,b_T} \approx {e^{-j(\gamma_{q,i}^R+\gamma_{q,j}^T)} \ G_{q} \ g_1(\theta_{AoD,q}, \theta^T_{b_T}) \ g_2(\phi_{AoD,q}, \theta_{AoD,q}, \phi^T_{b_T}, \theta^T_{b_T})}  + n_{i,j,l,b_T} \ .
\end{eqnarray}

\noindent{Case 3(b,2):} $(\theta^T_{b_T} = \theta_{AoD,q})$ and $(\phi^T_{b_T} \neq \phi_{AoD,q})$.
For large $N$, we get
\vspace{-2mm}
\begin{eqnarray}
h_{i,j,l,b_T} \approx {e^{-j(\gamma_{q,i}^R+\gamma_{q,j}^T)} \ G_{q} \ \sqrt{N} \ g_2(\phi_{AoD,q}, \theta_{AoD,q}, \phi^T_{b_T}, \theta_{AoD,q})}  + n_{i,j,l,b_T} \ .
\end{eqnarray}

\noindent{Case 3(c,1):} $(\theta^T_{b_T} = \theta_{AoD,k})$ and $(\phi^T_{b_T} = \phi_{AoD,k})$, with $k \neq q$.
\vspace{-2mm}
\begin{eqnarray}
h_{i,j,l,b_T} \approx {e^{-j(\gamma_{q,i}^R+\gamma_{q,j}^T)} \ G_{q} \ g_1(\theta_{AoD,q}, \theta_{AoD,k}) \ g_2(\phi_{AoD,q}, \theta_{AoD,q}, \phi_{AoD,k}, \theta_{AoD,k})} \\
+  {e^{-j(\gamma_{k,i}^R+\gamma_{q,k}^T)} \ G_{k} \ g_1(\theta_{AoA,k}, \theta_{AoA,q}) \ g_2(\phi_{AoA,k}, \theta_{AoA,k}, \phi_{AoA,q}, \theta_{AoA,q})}  + n_{i,j,l,b_T} \nonumber .
\end{eqnarray}

\noindent{Case 3(c,2):} $(\theta^T_{b_T} = \theta_{AoD,q})$ and $(\phi^T_{b_T} = \phi_{AoD,q})$.
\vspace{-2mm}
\begin{eqnarray}
h_{i,j,l,b_T} \approx e^{-j(\gamma_{q,i}^R+\gamma_{q,j}^T)} \ G_{q} \ N  + n_{i,j,l,b_T} \ .
\end{eqnarray}

Using the same ideas as in Case 1 and Case 2, and observing that the "signal" term contributed by Case 3(c,2) to $S_{i,j,l}$ dominates for large $N$ (it scales as $N^2$, as compared to the contribution of Case 3(b,2) (scales as $N$), and the contribution of other cases (independent of $N$)), we therefore get (assuming that there exists a Tx codebook beam that satisfies the condition for Case 3(c,2); cf. Remark 2, page 29)
\vspace{-2mm}
\begin{equation}\label{eq:P_R_Eff3}
P^R_{eff}(l) \approx \sigma^2 + N^2 |G_{q}|^2.
\end{equation}

The preceding equation proves part (b) of Lemma \ref{lemma:main}. To prove part (a), it remains to show that, for large $N$, the effective power in \eqref{eq:P_R_Eff3} dominates the effective power in \eqref{eq:P_R_Eff1} and \eqref{eq:P_R_Eff2}. This is evident based on the fact that the power in $\eqref{eq:P_R_Eff3}$ scales as $N^2$, while the power in \eqref{eq:P_R_Eff1} is independent of $N$ and the power in \eqref{eq:P_R_Eff2} scales as $N$. Consequently, by picking $N$ large enough, it is possible to ensure that the power in \eqref{eq:P_R_Eff3} is greater than that in \eqref{eq:P_R_Eff1} and \eqref{eq:P_R_Eff2} (with arbitrarily large probability). More precisely, consider the comparison between the power in \eqref{eq:P_R_Eff3} and \eqref{eq:P_R_Eff2}. We need to prove that, for large $N$, $N |G_{m}|^2 |g_2(\phi_{AoA,m}, \theta_{AoA,m}, \phi^R_l, \theta_{AoA,m})|^2 < N^2 |G_{q}|^2$, or, $\frac{|G_{m}|^2 |g_2(\phi_{AoA,m}, \theta_{AoA,m}, \phi^R_l, \theta_{AoA,m})|^2} {|G_{q}|^2} < N$. Replacing $g_2(\cdot)$ by the upper bound in \eqref{eq:ub2}, it suffices to show that, by picking $N$ large enough, we can ensure that
\vspace{-2mm}
\begin{equation}
\frac{2|G_{m}|^2} {|1-e^{j k d (sin(\theta_{AoA,m})sin(\phi_{AoA,m})-sin(\theta_{AoA,m})sin(\phi^R_l))}|^2 \ |G_{q}|^2} < N \ .
\end{equation}
\vspace{-2mm}
\noindent The expression on the LHS of this inequality is a random variable (the channel gains $G_m$ and $G_q$ are random variables), so that we are interested in studying the preceding inequality in a probabilistic sense. Further, note that this random variable, say $A$, is not a function of $N$. Therefore, denoting the cumulative distribution function of $A$ as $F_A(a)$, for $N > {F_A}^{-1}(1-\epsilon)$, $Pr(x<N) > 1-\epsilon$, $\forall \epsilon$. Hence, making $\epsilon$ arbitrarily small, the desired inequality holds with arbitrarily large probability. That the power in \eqref{eq:P_R_Eff3} exceeds the power in \eqref{eq:P_R_Eff1} can be shown in a similar manner, thereby completing the proof for part (a) of the Lemma.

{\emph{Remark 2:}} While computing the partial sum $S_{i,j,l}$ in Case (2) (Rx beam steered towards an elevation AoA) and Case (3) (Rx beam steered towards an elevation-azimuth AoA), we assumed existence of a Tx codebook beam in the direction of the corresponding AoDs. This simplifies the analysis. (If these assumptions are relaxed, certain conclusions may still be derived, although we do not conduct this analysis here). Note, however that, (in spirit of the discussion following the statement of Lemma 1 (page 13)), it is reasonable to expect that even though these assumptions may not hold strictly, the results of Lemma 1 can still be understood to hold in the approximate sense, assuming that the Tx codebook contains beams reasonably close to the AoD values.\looseness-1

\vspace{-3mm}


\end{document}